\documentclass[aps,pra,twocolumn,floatfix,nofootinbib,superscriptaddress,graphics,longbibliography]{revtex4-1}
\usepackage[usenames, dvipsnames]{xcolor}
\usepackage{hyperref}
\hypersetup{urlcolor=blue, colorlinks=true,citecolor=blue,linkcolor=blue} 
\usepackage[toc,page]{appendix}
\usepackage{bm}
\usepackage{graphicx}
\usepackage{thmtools} 
\usepackage{thm-restate}
\usepackage{mathrsfs}
\usepackage{amsmath}
\usepackage{amsthm}
\usepackage{amssymb}
\usepackage{mathtools} 
\usepackage{makecell}   
\usepackage{amssymb}
\usepackage{bbm, dsfont}
\usepackage{color}
\usepackage{times,txfonts}
\usepackage{nicefrac}
\usepackage{framed}
\usepackage{float}
\usepackage{tikz,pgfplots}
\usepackage[normalem]{ulem} %

\def\E{ {\cal E} }

\def\S{ {\cal S} }

\def\T{ {\cal T} }

\def\>{\rangle}
\def\<{\langle}

\newcommand{\bra}[1]{\langle {#1} |}
\newcommand{\ket}[1]{| {#1} \rangle}
 
\newcommand{\ketbra}[2]{\ensuremath{\left|#1\right\rangle\!\!\left\langle#2\right|}}

\newcommand{\tr}[1]{\mathrm{Tr}\left( #1 \right)}

\newcommand{\iden}{\mathbbm{1}}

\renewcommand{\v}[1]{\ensuremath{\boldsymbol #1}}

\newcommand{\rd}[1]{{\color{black}#1}}

\theoremstyle{plain}
\newtheorem{thm}{Theorem}
\newtheorem{lem}[thm]{Lemma}

\newtheorem{cor}[thm]{Corollary}
\newtheorem{obs}[thm]{Observation}
\newtheorem{defn}{Definition}

\begin{document}

\title{Geometric structure of thermal cones}
\date{\today}
\author{A. de Oliveira Junior}
\affiliation{Faculty of Physics, Astronomy and Applied Computer Science, Jagiellonian University, 30-348 Kraków, Poland.}
\author{Jakub Czartowski}
\affiliation{Faculty of Physics, Astronomy and Applied Computer Science, Jagiellonian University, 30-348 Kraków, Poland.}
\author{Karol {\. Z}yczkowski}
\affiliation{Faculty of Physics, Astronomy and Applied Computer Science, Jagiellonian University, 30-348 Kraków, Poland.}
\affiliation{Center for Theoretical Physics of the Polish Academy of Sciences Al. Lotników 32/46 02-668 Warsaw, Poland}
\author{Kamil Korzekwa}
\affiliation{Faculty of Physics, Astronomy and Applied Computer Science, Jagiellonian University, 30-348 Kraków, Poland.}

\begin{abstract}
The second law of thermodynamics imposes a fundamental asymmetry in the flow of events. The so-called thermodynamic arrow of time introduces an ordering that divides the system's state space into past, future and incomparable regions. In this work, we analyse the structure of the resulting thermal cones, i.e., sets of states that a given state can thermodynamically evolve to (the future thermal cone) or evolve from (the past thermal cone). Specifically, for a $d$-dimensional classical state of a system interacting with a heat bath, we find explicit construction of the past thermal cone and the incomparable region. Moreover, we provide a detailed analysis of their behaviour based on thermodynamic monotones given by the volumes of thermal cones. Results obtained apply also to other majorisation-based resource theories (such as that of entanglement and coherence), since the partial ordering describing allowed state transformations is then the opposite of the thermodynamic order in the infinite temperature limit. \rd{Finally, we also generalise the construction of thermal cones to account for probabilistic transformations.} 
\end{abstract}

\maketitle


\section{Introduction}
\label{Sec:Introduction}

Thermodynamic evolution of physical systems obeys a fundamental asymmetry imposed by nature. Known as the thermodynamic arrow of time~\cite{eddington1928nature}, it is a direct manifestation of the second law of thermodynamics, which states that the entropy of an isolated system cannot decrease~\cite{fermi1956thermodynamics, callen1985thermodynamics}. In other words, the thermodynamic evolution inherently distinguishes the past from the future: systems spontaneously evolve to future equilibrium states, but do not spontaneously evolve away from them. Even though recognition of the thermodynamic arrow of time is an old discussion~\cite{boltzmann1895certain,zermelo1896satz}, it still raises deep questions relevant both to philosophy and the foundations of physics~\cite{prigogine2000arrow, price2004origins}. Despite many attempts, the full understanding of the time asymmetry in thermodynamics seems to be still beyond our reach. In recent years, we have seen renewed interest in exploring the thermodynamics of a few quanta, which motivated the development of a powerful theoretical toolkit~\cite{Goold2016, Deffner2019,Landi2021,vom2019d} that allows one to revisit old questions.

The toolkit in question, providing a robust approach to study thermodynamics of small systems, is given by the resource-theoretic framework~\cite{Janzing2000,horodecki2013quantumness, brandao2015second,Lostaglio2019}. It not only recovers the macroscopic results, but also presents a suitable platform to address relevant problems within the field of
thermodynamics, including a rigorous derivation of the third law~\cite{Wilming2017} and study of cooling mechanisms~\cite{Lostaglio2018elementarythermal,Lostaglio2019}. Furthermore, it provides a perfect setup for modelling thermalisation in many-body physics~\cite{Sparaciari2021}. Regarding the thermodynamic arrow of time, a resource-theoretic analysis was performed in Ref.~\cite{Korzekwa2017}, where the author investigated this problem from the point of view of order theory. However, Ref.~\cite{Korzekwa2017} focused on structural differences between classical and quantum theories in contrast to the geometric aspects of thermal cones that we investigate here.

In this work we aim at characterising the thermodynamic arrow of time by investigating allowed transformations between energy-incoherent states that arise from the most general energy-conserving interaction between the system and a thermal bath. These transformations encode the structure of the thermodynamic arrow of time by telling us which states can be reached from a given state, here called the \emph{present state}, in accordance with the laws of thermodynamics. Under these constraints, the state space can then be naturally decomposed into three parts: the set of states to which the present state can evolve is called the \emph{future thermal cone}; the set of states that can evolve to the present state is called the \emph{past thermal cone}; while states that are neither in the past nor the future thermal cone form the \emph{incomparable region}.

Although substantial insights have been obtained by studying the future thermal cone~\cite{Lostaglio2018elementarythermal,mazurek2018decomposability,mazurek2019thermal}, explicit characterisation of the incomparable region and the past thermal cone has not been performed. Our main results consist of two theorems addressing this gap. The first one yields a geometric characterisation of the past majorisation cone (i.e., the set of probability distributions majorising a given distribution) by means of explicit construction, which, together with the incomparable region and future thermal cone, fully specifies the time-like ordering in the probability simplex in the limit of infinite temperature. The second result, derived using a novel tool of an embedding lattice, generalises the first one to the case of finite temperatures. \rd{Similar results are also obtained for probabilistic scenarios~\cite{vidal1999entanglement}.} The approach presented here allows us to conduct a rigorous study of the causal structure of the thermal cones. 
It is known from earlier works that the future (thermal) cone is convex~\cite{bhatia1996matrix,Lostaglio2018elementarythermal}; our results extend this knowledge by showing rigorously that the past (thermal) cone can always be decomposed into $d!$ convex parts and, in the zero-temperature limit, only one of them retains a non-zero volume, rendering the entire past thermal cone convex. We also introduce new thermodynamic monotones given by the volumes of the past and future thermal cones.

Our results can also be seen as an extension of the famous Hardy-Littlewood-P{\'o}lya theorem~\cite{hardy1952inequalities}, as they specify the past cone and the incomparable region in addition to the previously studied future cone. Therefore, they can also be employed to study other majorisation-based resource theories, such as the theory of entanglement~\cite{nielsen1999conditions} or coherence~\cite{Plenio2014, Du2015}. Concerning local operations, an analogy between special relativity and the set of pure states of bi-partite systems was previously made in Ref.~\cite{zyczkowski2002}, where the authors correspondingly divided the state space into three parts. Here, we consider a more general partial-order structure, the thermomajorisation order, which generalises the previous and recovers it in the limit of infinite temperature.

The paper is organised as follows. We set the scene in Sec.~\ref{Sec_Setting-the-scene} by recalling the resource-theoretic approach to thermodynamics and introducing the necessary concepts. We also collect there known results concerning the conditions for state transformation under majorisation and thermomajorisation. In Sec.~\ref{Sec:Majorisation-cones}, we state the first of our main results concerning the construction of majorisation cones and discuss its interpretation within the thermodynamic setting and in other majorisation-based theories. \rd{This construction is also generalised for probabilistic transformations}. Sec.~\ref{Sec:Thermal-cones} is devoted to the analogous study of the thermal cones, generated by thermomajorisation relation, where we also introduce the  tool of embedding lattice, instrumental for the proof of the second main result. Sec.~\ref{Sec_volume_thermal_cones}, we introduce thermodynamic monotones given by the volumes of the past and future thermal cone, discuss their intuitive operational interpretation and describe their properties. \rd{We also comment on the different natures of future and past cones for entanglement transformations.} Finally, we conclude with an outlook in Sec.~\ref{Sec:Outlook}. The technical derivation of the main results can be found in Appendix~\ref{sec_appendix}. \rd{In Appendix~\ref{app:probal_deriv}, we derive the analogue of the majorisation cones for probabilistic transformations. Next, Appendix~\ref{app:entanglement_vols} presents methods used to obtain insights into the volumes of majorisation cones in the context of quantum entanglement. Appendix~\ref{app:coherent_thermal} shows how to obtain the coherent equivalents of the future and past cones together with the incomparable regions for coherent states under thermal operations~\cite{horodecki2013fundamental} and Gibbs-preserving operations~\cite{Faist_2015}.}


\section{Setting the scene} \label{Sec_Setting-the-scene}


\subsection{Thermodynamic evolution of energy-incoherent states}

In what follows, we study the thermodynamics of finite-dimensional quantum systems in the presence of a heat bath at temperature $T$. The system under investigation is described by a Hamiltonian \mbox{$H = \sum_i E_i \ketbra{E_i}{E_i}$}, and it is prepared in a state $\rho$; while the heat bath, with a Hamiltonian $H_E$, is in a thermal equilibrium state,
\begin{equation}
\label{eq_gibbs_state}
    \gamma_E=\frac{e^{-\beta H_E}}{\tr{e^{-\beta H_E}}},
\end{equation}
where $\beta=1/k_B T$ is the inverse temperature with $k_B$ denoting the Boltzmann constant. The evolution of the joint system is assumed to be closed, and therefore described by a unitary operator $U$, which is additionally enforced to conserve the total energy,
\begin{equation}
\label{eq_energy_conservation}
    [U, H\otimes \iden_E+ \iden\otimes H_E] = 0.
\end{equation} 

More formally, the set of allowed thermodynamic transformations is modelled by \emph{thermal operations}~\cite{Janzing2000, brandao2015second, horodecki2013fundamental}, which consists of completely positive trace-preserving (CPTP) maps that act on a system $\rho$ with Hamiltonian $H$ as
\begin{equation}
    \label{eq:thermal_ops}
    \E(\rho)=\textrm{Tr}_E\left(U\left(\rho\otimes\gamma_E\right)U^{\dagger}\right),
\end{equation}
with $U$ satisfying Eq.~\eqref{eq_energy_conservation} and the state $\gamma_E$ given by Eq.~\eqref{eq_gibbs_state} with an arbitrary Hamiltonian $H_E$. Note that the energy conservation condition, Eq.~\eqref{eq_energy_conservation}, can be interpreted as encoding the first law of thermodynamics. Moreover, the fact that the heat bath is in thermal equilibrium implies that every thermal operation $\E$ is a map which preserves the Gibbs state,
\begin{equation}
    \label{eq_Gibbs_fixedpoint}
\E(\gamma)=\gamma , 
\end{equation}
with $\gamma$ being the thermal Gibbs state of the system given by Eq.~\eqref{eq_gibbs_state} with $H_E$ replaced by $H$. Thus, Eq.~\eqref{eq_Gibbs_fixedpoint} combined with a contractive distance measure $\delta$ resulting in \mbox{$\delta(\rho,\gamma)\geq\delta(\E(\rho),\E(\gamma))=\delta(\E(\rho),\gamma)$}, incorporates the core physical principle of the second law of thermodynamics as it captures the idea of evolution towards thermal equilibrium.

In this work we focus on \emph{energy-incoherent states}, i.e., quantum states $\rho$ that commute with the Hamiltonian~$H$. \rd{This restriction can be justified in the quasi-classical regime, in which systems are quantised and possess only a small number of energy levels, but the decoherence is so strong that the interference effects between the different energy levels become negligible. Energy-incoherent} states can be equivalently represented by a $d$-dimensional probability vectors $\v p$ of their eigenvalues, which coincide with populations in the energy eigenbasis.  Note that $\gamma$ commutes with $H$ and thus can be represented by a probability vector

\begin{equation}
\label{eq_thermal-distribution}
    \v \gamma = \frac{1}{Z}\left(e^{-\beta E_1},..., e^{-\beta E_d}\right) \quad \text{where} \quad Z = \sum_i e^{-\beta E_i} .
\end{equation}
\rd{Here, one should note the difference between the state $\gamma$ and the probability vector $\v \gamma$.} The most general evolution between two probability vectors, $\v{p}$ and $\v{q}$, is described by a stochastic matrix~$\Lambda$, where $\Lambda$ satisfies $\Lambda_{ij} \geq 0$ and $\sum_i \Lambda_{ij}=1$. In the thermodynamic context, however, we may restrict the analysis to stochastic matrices preserving the Gibbs state,
i.e., to $\Lambda$ such that $\Lambda \v{\gamma}=\v{\gamma}$, in short called Gibbs-preserving (GP) matrices in the literature~ \cite{horodecki2013fundamental,Faist_2015}. 
This is because of the following theorem linking the existence of a thermal operation between incoherent states to the existence of a Gibbs-preserving stochastic matrix between probability distributions representing these states~\cite{Janzing2000,horodecki2013fundamental,Korzekwathesis}.
\begin{thm}[Theorem~5 of Ref.~\cite{Janzing2000}]
	\label{thm_Thermaloperations}
    Let $\rho$ and $\sigma$ be quantum states commuting with the system Hamiltonian $H$, and $\gamma$ its thermal Gibbs state with respect to the inverse temperature $\beta$. Denote their eigenvalues by $\v{p}$, $\v{q}$ and $\v{\gamma}$, respectively. Then, there exists a thermal operation $\E$, such that $\E(\rho)=\sigma$, if and only if there exists a stochastic map $\Lambda$ such that
		\begin{equation}
		    \Lambda\v{\gamma}=\v{\gamma}~~\mathrm{and}~~\Lambda \v p=\v q . 
		\end{equation}
\end{thm}

As a result of the above theorem, while studying thermodynamic transformations between energy-incoherent states, we can replace CPTP maps and density matrices with stochastic matrices and probability vectors, respectively. Then, the rules describing what state transformations are allowed under thermal operations, can be expressed by a partial-order relation between probability vectors corresponding to initial and final states. In the infinite-temperature limit, these rules are encoded by the majorisation relation~\cite{marshall1979inequalities}, and in the finite-temperature case by thermomajorisation~\cite{Rusch,horodecki2013fundamental}. 


\subsection{Infinite temperature and the majorisation order}

\begin{figure}[t]
    \centering
    \includegraphics{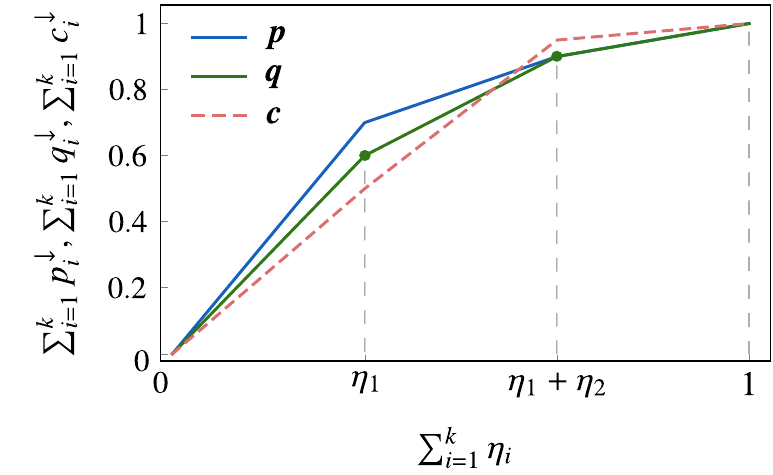}
    \caption{\label{fig-majorisationcurves} \emph{Majorisation curve.} For three different states $\v p$, $\v q$ and $\v c$, we plot their majorisation curves $f_{\v p}(x), f_{\v q}(x)$ and $f_{\v c}(x)$, respectively. While $\v p$ majorises $\v q$ (since $f_{\v p}(x)$ is never below $f_{\v q}(x)$), both states are incomparable with $\v c$, as their majorisation curves cross with $f_{\v c}(x)$.}  
\end{figure}

In the infinite temperature case, $\beta = 0$, the thermal Gibbs state $\v \gamma$ is described by a uniform distribution
\begin{equation}
    \label{eq_uniform-state}
    \v \eta = \frac{1}{d}(1, ..., 1).
\end{equation}
Theorem~\ref{thm_Thermaloperations} implies that a state $\v{p}$ can be mapped to $\v{q}$ if and only if there exists a bistochastic matrix ($\Lambda$, such that $\Lambda \v{\eta}=\v{\eta}$), which transforms $\v{p}$ into $\v{q}$. To formulate the solution we need to recall the concept of majorisation~\cite{marshall1979inequalities}:
\begin{defn}[Majorisation]\label{def_Majorisation} Given two $d$-dimensional probability distributions $\v p$ and $\v q$, we say that $\v{p}$ \emph{majorises} $\v{q}$, and denote it by $\v p \succ \v q$, if and only if
\begin{equation}
\label{eq_majorisation}
    \sum_{i=1}^k p_i^{\downarrow}\geq\sum_{i=1}^k q_i^{\downarrow} \quad \text{for all} \quad  k\in\{1\dots d\},
\end{equation}
where $\v{p}^{\downarrow}$ denotes the vector $\v{p}$ rearranged in a non-increasing order. 
\end{defn}

Equivalently, the majorisation relation can be expressed in a more geometric way by defining a \emph{majorisation curve}, i.e., a piecewise linear curve $f_{\v p}(x)$ in $\mathbb{R}^2$ obtained by joining the origin $(0,0)$ and the points $\left(\sum_{i=1}^{k} \eta_i, \sum_{i=1}^{k} p^{\downarrow}_i \right)$, for $k \in \{1, ..., d\}$. Then, a distribution $\v{p}$ majorises $\v{q}$ if and only if the majorisation curve $f_{\v p}(x)$ of $\v{p}$ is always above that of $\v{q}$,
\begin{equation}
\v p \succ \v q \iff \forall x\in \left[0,1\right]:~f_{\v p}(x) \geq  f_{\v q}(x) \, .
\end{equation}
Since majorisation does not introduce a total order, a given pair of states $\v p$ and $\v q$ may be incomparable with each other: neither $\v{p}$ majorises~$\v{q}$, nor $\v{q}$ majorises $\v{p}$. In terms of majorisation curves, this implies that both curves intersect each other~ (see Fig.~\ref{fig-majorisationcurves}). 

We are now ready to state the connection between majorisation relation and bistochastic state transformations captured by the celebrated Hardy-Littlewood-P{\'o}lya theorem~\cite{hardy1952inequalities}.
\begin{thm}[Theorem~II.1.10 of Ref.~\cite{bhatia1996matrix}]
\label{thm_HLP}
	There exists a bistochastic matrix $\Lambda$, $\Lambda \v \eta=\v \eta$, mapping $\v{p}$ to $\v{q}$ if and only if $\v{p} \succ \v{q}$.
\end{thm}

Furthermore, it is important to mention that the majorisation order forms a lattice structure~\cite{cicalese2002}, within which there exists a natural time-like hierarchy of elements.

\begin{defn}[Lattice]\label{def_Lattice} A partially ordered set $(L, \leq)$ forms a lattice if for every pair of elements $\v p ,\v q \in L$, there exists a least upper bound, called \textbf{join} and denoted by $\v p \vee \v q$, such that $\v p \vee \v q \geq \v p$ and $\v p \vee \v q \geq \v q$; and a greatest lower bound, called \textbf{meet} and denoted by $\v p \wedge \v q$, such that $\v p \wedge\v q \leq \v p$ and $\v p \wedge \v q \leq \v q$.
\end{defn}

Viewing the lattice structure of majorisation order from the perspective of the laws governing state transformation under bistochastic matrices, it becomes clear that a join can be understood as the last point of the common ``past'' for a given two elements of the lattice, and similarly, a meet can be seen as the first point of the common future~\cite{Korzekwa2017}. The existence of both join and meet within the majorisation order lattice has been proven in Ref.~\cite{cicalese2002} and will be instrumental in constructing the past and incomparable regions.


\subsection{Finite temperature and thermomajorisation order}

Theorem~\ref{thm_HLP} can be generalised from bistochastic matrices to Gibbs-preserving matrices, and the aim of this section is to present such a result without focusing on its derivation (for details see Ref.~\cite{Korzekwathesis}). To achieve this, we first need to introduce the thermodynamic equivalent of the majorisation partial order, which can be achieved by first presenting the concept of $\beta$-ordering, and then introducing the notion of thermomajorisation curves~\cite{horodecki2013fundamental}. For a given initial state $\v p$ and a thermal Gibbs distribution \mbox{$\v \gamma$} with inverse temperature $\beta$, the $\beta$-ordering of $\v p$ is defined by a permutation $\v \pi_{\v p}$ that sorts $p_i/\gamma_i$ in a non-increasing order,
\begin{equation} \label{eq_beta-ordering}
\frac{p_{\v \pi^{-1}_{\v p}(i)}}{\gamma_{\v \pi^{-1}_{\v p}(i)}} \geq \frac{p_{\v \pi^{-1}_{\v p}(i+1)}}{\gamma_{\v \pi_{\v p}^{-1}(i+1)}} \:\:, \:\: \text{for} \:\: i \in \{1, ..., d-1\} .
\end{equation}
Thus, the $\beta$-ordered version of $\v p$ is given by
	\begin{equation}
	\v{p}^{\, \beta}=\left(p_{\v \pi_{\v{p}}^{-1}(1)},\dots ,p_{\v \pi_{\v{p}}^{-1}(d)}\right).
	\end{equation}
Note that each permutation belonging to the symmetric group, $\v \pi \in \mathcal S_d$, defines a different $\beta$-ordering on the energy levels of the Hamiltonian $H$.

A thermomajorisation curve \mbox{$f^{\, \beta}_{\v{p}}:\left[0,1\right]\rightarrow\left[0,1\right]$} is defined as a piecewise linear curve composed of linear segments connecting the point $(0,0)$ and the points defined by consecutive subsums of the $\beta$-ordered form of the probability $\v{p}^\beta$ and the Gibbs state $\v{\gamma}^\beta$,
\begin{equation}
    \left(\sum_{i=1}^k\gamma^{\, \beta}_i,~\sum_{i=1}^k p^{\, \beta}_i\right):=\left(\sum_{i=1}^k\gamma_{\v \pi^{-1}_{\v{p}}(i)},~\sum_{i=1}^k p_{\v \pi^{-1}_{\v{p}}(i)}\right),
\end{equation}
for $k\in\{1,\dots,d\}$. 

Finally, given two $d$-dimensional probability distributions $\v p$ and $\v q$, and a fixed inverse temperature $\beta$, we say that $\v p$ \emph{thermomajorises} $\v q$ and denote it as $\v p \succ_{\beta} \v q$, if the thermomajorisation curve $f^{\, \beta}_{\v{p}}$ is above $f^{\, \beta}_{\v{q}}$ everywhere, i.e.,
\begin{equation}
    \v p \succ_{\beta} \v q \iff \forall x\in[0,1]:~ f^{\, \beta}_{\v{p}}(x) \geq f^{\, \beta}_{\v{q}}(x) \, .
\end{equation}
As before, it may happen that given two vectors, $\v p$ and $\v q$, are incomparable, meaning that $f^{\, \beta}_{\v p}$ and $f^{\, \beta}_{\v q}$ cross at a some point (see Fig.~\ref{fig-thermomajorisationcurves}). Furthermore, the thermomajorisation order cannot be seen as a lattice in the sense that for $\v p$ and $\v q$, which do not share the same $\beta$-order, there is no unique join or meet~\cite{Korzekwa2017}. 

We can now state the generalisation of Theorem~\ref{thm_HLP} from bistochastic matrices to Gibbs-preserving matrices which, via Theorem~\ref{thm_Thermaloperations}, yields the necessary and sufficient conditions for the existence of a thermal operation between two energy-incoherent states~\cite{horodecki2013fundamental,Rusch1978}.
\begin{thm}[Theorem~1.3 of Ref.~\cite{Janzing2000}]
	\label{thm_HLPgeneralisation}
    There exists a Gibbs-preserving matrix $\Lambda$, $\Lambda \v \gamma=\v \gamma$, mapping $\v{p}$ to $\v{q}$ if and only if $\v{p} \succ_{\beta} \v{q}$.
\end{thm}

\begin{figure}[t]
    \centering
    \includegraphics{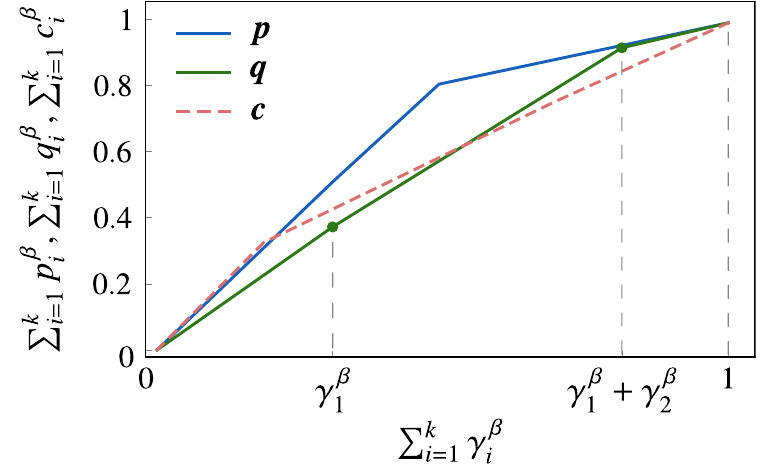}
    \caption{\label{fig-thermomajorisationcurves} \emph{Thermomajorisation curve}. For three different states $\v p$, $\v q$ and $\v c$, with a thermal Gibbs state $\v \gamma = (1, e^{-\beta}, e^{-2\beta})$, and $\beta >0$, we plot their thermomajorisation curves $f^{\, \beta}_{\v p}(x), f^{\, \beta}_{\v q}(x)$ and $f^{\, \beta}_{\v c}(x)$, respectively. While $\v p$ thermomajorises $\v q$ (since $f^{\, \beta}_{\v p}(x)$ is never below $f^{\, \beta}_{\v q}(x)$), both states are incomparable with $\v c$, as their thermomajorisation curves cross with $f^{\, \beta}_{\v c}(x)$.} 
\end{figure}


\section{Majorisation cones} 
\label{Sec:Majorisation-cones}

The reachability of states under bistochastic matrices can be studied by introducing the notion of \emph{majorisation cones}, defined as follows:

\begin{defn}[Majorisation cones]\label{def_MajorisationCones}
The set of states that a probability vector $\v p$ can be mapped to by bistochastic matrices is called the \textbf{future cone} $\T_+(\v{p})$. The set of states that can be mapped to $\v p$ by bistochastic matrices is called the \textbf{past cone} $\T_-(\v{p})$. The set of states that are neither in the past nor in the future cone of $\v p$ is called the \textbf{incomparable region} $\T_{\emptyset}(\v{p})$.
\end{defn}

Definition~\ref{def_MajorisationCones} provides us with a stage for studying thermodynamics of energy-incoherent states in the infinite-temperature limit. What is more, it also allows us to explore a more general class of state transformations whose rules are based on a majorisation relation, e.g., the resource theories of entanglement~\cite{nielsen1999conditions} or coherence~\cite{Plenio2014,Du2015, Streltsov2016,Streltsov2017}. These more general settings are discussed in more detail in Sec.~\ref{SubSec_LinkstootherQuantumTheorioes}, whereas now we focus on thermodynamic transformations. 


\subsection{Geometry of majorisation cones}

The future cone can be easily characterised by employing the Birkhoff’s theorem~\cite{bhatia1996matrix} stating that every bistochastic matrix can be written as a convex combination of permutation matrices. Thus, the set of $d \times d$ bistochastic matrices is a convex polytope with $d!$ vertices, one for each permutation in $\mathcal{S}_d$. Combining this observation with Theorem~\ref{thm_HLP}, we obtain the future cone of~$\v p$:
\begin{cor}[Future cone]
\label{thm_futureinfinite}
For a $d$-dimensional probability vector $\v{p}$, its future cone is given by
\begin{equation}
    \label{eq-futurecone}
    \T_{+}(\v p) = \operatorname{conv}\left[\left\{\Pi \v p \, , \mathcal{S}_d \ni \v \pi \mapsto \Pi  \right\}\right] ,
\end{equation}
where $\Pi$ denotes a permutation matrix corresponding to the permutation $\v \pi$ with $d$ elements, and $\operatorname{conv[S]}$ the convex hull of the set $S$.
\end{cor}

The above corollary implies that the future cone of $\v p$ is a convex set with all distributions lying in $\T_{+}(\v p)$ being majorised by $\v p$. Since the $d$-dimensional sharp distribution, $(0,..., 1,.., 0)$, majorises all probability distributions, its future cone is the entire probability simplex, which we will denote by $\v \Delta_d$. For later convenience, we will introduce the concept of $d!$ Weyl chambers, each composed of probability vectors which can be ordered non-decreasingly by a common permutation $\v{\pi}$. In particular, we will refer to the chamber corresponding to the identity permutation as the \emph{canonical Weyl chamber}.

If there is no transformation mapping $\v p$ into $\v q$ nor $\v q$ into~$\v p$, we say that these two states are incomparable. The incomparable region can be characterised by incorporating into the analysis the concept of quasi-probability distributions, which are defined by relaxing the non-negativity condition on the entries of a normalised probability distribution. The following result, the proof of which is employing the lattice structure of majorisation order and can be found in Appendix~\ref{sec_appendix}, specifies the incomparable region of $\v p$.
\begin{lem}[Incomparable region]
\label{lem_incomparablecone}
For a $d$-dimensional probability distribution \mbox{$\v{p} = (p_1,..., p_d)$}, consider the quasi-probability distributions $\v{t}^{(n)}$ constructed for each \mbox{$n \in \{1, ..., d\}$},
\begin{equation}
\label{eq_tangentvectors}
    \v{t}^{(n)}=\left(t^{(n)}_1, p_n^{\downarrow}, ..., p_n^{\downarrow}, t^{(n)}_d\right) , 
\end{equation}
with 
\begin{align}
\label{eq:extremalpointspastan}
    t_1^{(n)} & = \sum_{i=1}^{n-1} p_i^{\downarrow} - (n-2) p_n^{\downarrow}, &
    t_d^{(n)} & = 1 - t^{(n)}_1 - (d-2)p_n^{\downarrow},
\end{align}
and define the following set
\begin{equation}
\label{eq:setP}
    \mathbb{T}:=\bigcup\limits_{j=1}^{d-1} \operatorname{conv}\left[\T_{+}\left(\v{t}^{(j)}\right) \cup \T_{+}\left(\v{t}^{(j+1)}\right)\right] .    
\end{equation}
Then, the incomparable region of $\v{p}$ is given by 
\begin{align}
\label{eq:incomparableconeIT}
    \T_{\emptyset}(\v{p}) = \left[\operatorname{int}(\mathbb{T})\backslash\T_+\left(\v p\right)\right]\cap \v{\Delta}_d  ,
\end{align}
where the backslash $\backslash$ denotes the set difference and $\operatorname{int}( \mathbb{T})$ represents the interior of  $\mathbb{T}$.
\end{lem}

\begin{figure}[t]
    \centering
    \includegraphics{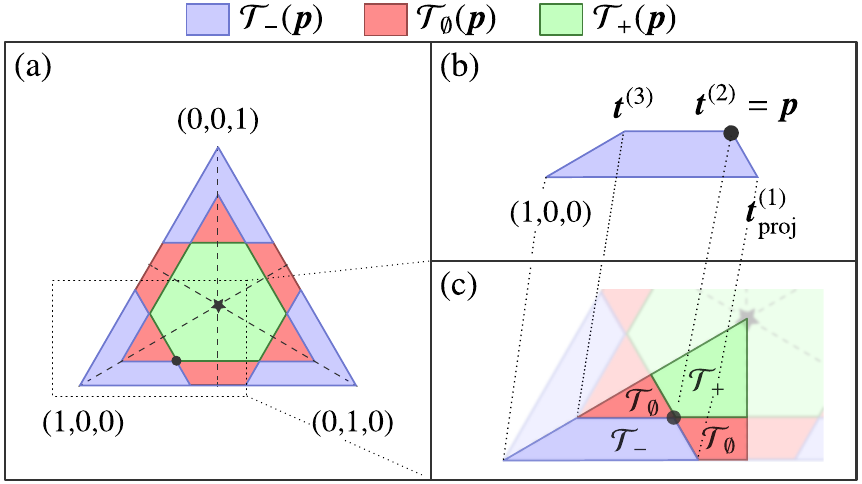}
    \caption{\label{fig-past-chamber} \emph{Majorisation cones and Weyl chambers}. (a) Probability simplex $\Delta_3$ and a state $\v p = (0.6,0.3,0.1)$ \rd{represented by a black dot~$\bullet$} together with its majorisation cones. The division of $\Delta_3$ into different Weyl chambers is indicated by dashed lines \rd{with the central state \mbox{$\v \eta =(1/3,1/3,1/3)$} denoted by a black star~$\bigstar$}. (b) The past cone of a state~$\v{p}$ restricted to a given Weyl chamber is convex with the extreme points given by $\v{t}^{(n)}$ from Eq.~\eqref{eq_tangentvectors} and the sharp state. (c) The causal structure induced by bistochastic matrices (i.e., thermal operations in the infinite temperature limit) in a given Weyl chamber.}  
\end{figure}

We will refer to the quasi-probability distributions $\v{t}^{(n)}$ as \emph{tangent vectors}. The intuition behind this name and the importance of $\v{t}^{(n)}$ can be explained by noticing that any convex function $g(x)$ lies fully under its tangent at any point $y$, denoted as $t_y(x) \geq g(x)$, with equality guaranteed only for \mbox{$t_y(y) = g(y)$}. It follows from the definition of $\v t^{(n)}$ that its majorisation curve $f_{\v t^{(n)}}(x)$ is parallel to the $n$-th linear piece of $f_{\v p}(x)$  for $x\in\left[(n-1)/d,\,n/d\right]$, and the first and last elements of $\v t^{(n)}$ guarantee tangency and normalisation. Finally, since the adjacent linear fragments of $f_{\v p}(x)$ share the elbows of the function, the consequent tangent vectors $\v t^{(n)},\,\v t^{(n+1)}$ are both tangent at a selected elbow, $f_{\v p}(n/d) = f_{\v t^{(n)}}(n/d) = f_{\v t^{(n+1)}}(n/d)$. Therefore, any convex combination of the form $a \v t^{(n)} + (1-a)\v t^{(n+1)}$ will be ''tangent'' at the $n$-th elbow of the $\v p$ majorisation curve. The fact that $\v{t}^{(n)}$ may be a quasi-probability distribution does not pose a problem, since this vector can be projected back onto the probability simplex. The projected vector $\v{t}^{(n)}$ will be denoted by $\v{t}^{(n)}_{\text{proj}}$, and can be obtained by successively applying the map 
\begin{equation}
\left\{t^{(n)}_{m-1},\,t^{(n)}_m\right\} \longmapsto \left\{\min\left(t^{(n)}_{m-1}+t^{(n)}_m,\,t^{(n)}_{m-1}\right),\,\max\left(t^{(n)}_m,\,0\right)\right\}    
\end{equation}
to pairs of entries of $\v{t}^{(n)}$ going from $m = d$ to $m = 2$. In each step, the map either zeros the second component by shifting its value to the first one or, if the second component is non-negative, it leaves them both unperturbed. Geometrically, the state is shifted along the edges of the future cone of $\v{t}^{(n)}$ and every time it hits a plane defining one of the faces of the probability simplex $\Delta_d$, a new direction is selected, until the state is composed exclusively of non-negative entries.

Using Lemma~\ref{lem_incomparablecone}, we can now prove the following theorem that specifies the past cone.
\begin{thm}[Past cone]
\label{thm_pastcone}
The past cone of $\v{p}$ is given by
\begin{equation}
        \T_{-}(\v{p}) = \v{\Delta}_d \backslash \operatorname{int} (\mathbb{T}) \, .
\end{equation}
\end{thm}
\begin{proof}
    One only needs to use the fact that
    \begin{equation}
        \label{eq:pastconeIT}
        \T_{-}(\v{p}) = \v{\Delta}_d \texttt{\textbackslash} \left(\T_{\emptyset}(\v{p})\cup \T_{+}(\v{p})\right)  \, ,
\end{equation}
and employ Lemma~\ref{lem_incomparablecone} to replace $\T_{\emptyset}(\v{p})$ in the above with Eq.~\eqref{eq:incomparableconeIT}.
\end{proof}

Let us make a few comments on the above results. First, note that the incomparable region arises only for $d \geq 3$. This can be easily deduced from Lemma~\ref{lem_incomparablecone}, as for $d=2$ the two extreme points, $\v t^{(1)}$ and $\v t^{(2)}$, are precisely the initial state~$\v{p}$. Second, the future thermal cone is symmetric with respect to the maximally mixed distribution $\v \eta$, and consequently, the incomparable and past cones also exhibit a particular symmetry around this point. As we shall see, this symmetry is lost when we go beyond the limit of infinite temperature. Third, although the past cone is not convex as a whole, it is convex when restricted to any single Weyl chamber. 
Therefore, we may note that the tangent vectors $\v t^{(n)}$ provide the extreme points of the past not only from the viewpoint of a single-chamber but also to the entire probability simplex. This can be understood by noting that $\v t^{(n)}$ are located at the boundary between the incomparable and the past cone, and by symmetry, it applies to all their permuted versions. As a consequence, the past is constructed from $d!$ copies of the past in the canonical Weyl chamber, each copy transformed according to the corresponding permutation $\v{\pi}$~(see Fig.~\hyperref[fig-past-chamber]{\ref{fig-past-chamber}a, b}). 
Finally, one can make an analogy to special relativity with bistochastic matrices imposing a causal structure in the probability simplex $\Delta_d$. There exists a ``light cone'' for each point in $\Delta_d$, which divides the space into past, incomparable, and future regions~(see Fig.~\hyperref[fig-past-chamber]{\ref{fig-past-chamber}c}).
\begin{figure}[t]
    \centering
    \includegraphics{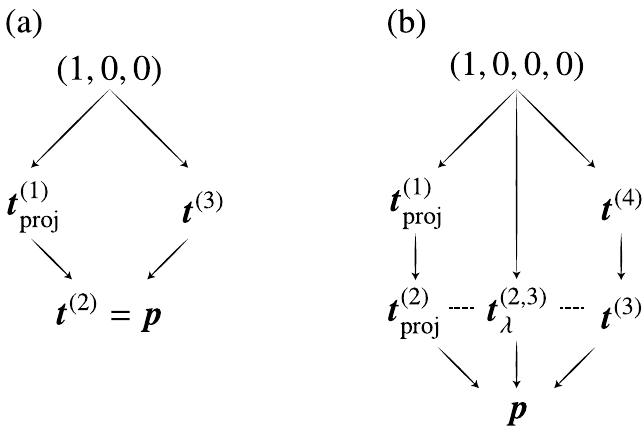}
    \caption{\label{fig-lattice-diagram} \emph{Partial-order diagrams for majorisation}. Graphical representation of the partially ordered set formed by the extreme points of the past cone in the probability simplex $\Delta_d$. Each arrow indicates that one of the elements precedes the other in the majorisation ordering. (a) Diagram for $d=3$. (b) Diagram for $d=4$. Note that any convex combination between $\v{t}^{(2)}_{\text{proj}}$ and $\v t^{(3)}$, here denoted by $\v{t}^{(2,3)}_{\lambda}$, results in an incomparable vector with respect to $\v{t}^{(1)}_{\text{proj}}$ and $\v{t}^{(4)}$.
   }
\end{figure}

The central idea behind Lemma~\ref{lem_incomparablecone} and Theorem~\ref{thm_pastcone} can be better understood through a visualisation using partial-order diagrams. To illustrate the principles of such diagrams, we will first focus on the special case of a three-level system. Then, the past cone has three non-trivial extreme points: $\v t^{(1)}$, $\v t^{(2)}$ and $\v t^{(3)}$. Furthermore, as shown in Fig.~\hyperref[fig-lattice-diagram]{\ref{fig-lattice-diagram}a}, these extreme points satisfy the following partial-order relation: the sharp state $\v s_1$ majorises both $\v t^{(1)}$ and $\v t^{(3)}$, and these two vectors majorise the initial state $\v p=\v{t}^{(2)}$. As it was proved in Lemma~\ref{lem_incomparablecone}, the union of the future cones of these extreme points provides us, after subtracting the future of the vector $\v{p}$, with the incomparable region of $\v p$. In the particular case of $d=3$ since $\v{t}^{(1)},\,\v{t}^{(3)}\succ\v{t}^{(2)}$, we find that $\operatorname{conv}[\mathcal{T}_+(\v{t}^{(1)}),\mathcal{T}_+(\v{t}^{(2)})] = \mathcal{T}_+(\v{t}^{(1)})$ and similarly for $\v{t}^{(3)}$. However, it is important to note the fact that $\v t^{(1)}$ and $\v t^{(3)}$ are incomparable, and in turn, their respective future cones after subtracting future of $\v{p}$ characterise disjoint parts of the incomparable region. Finally, the tangent vector $\v t^{(2)}$ reduces to the original probability vector, $\v t^{(2)} = \v p$, only for $d = 3$, and this fact is fully understood from the construction of the $\v t^{(n)}$-vectors (see Appendix~\ref{sec_appendix}). 

It is evident from Lemma~\ref{lem_incomparablecone} that each pair of tangent vectors $(\v t^{(n)},\,\v t^{(n+1)})$ characterises a given part of the incomparable region. However, notice that the futures of the extreme points considered one by one do not give the full description of the incomparable region -- one needs to consider their convex hulls to fill in the gaps. This particular feature of the construction can be demonstrated by considering the case of $d = 4$. Observe that in this case, we have a set of four tangent vectors $\v t^{(n)}$ with $n\in\{1,2,3,4\}$. Straightforward calculation shows that $\v t^{(1)}$ majorises $\v t^{(2)}$ and $\v t^{(4)}$ majorises $\v t^{(3)}$, therefore we find certain simplification, namely $\operatorname{conv}[\mathcal{T}_+(\v{t}^{(1)}),\mathcal{T}_+(\v{t}^{(2)})] = \mathcal{T}_+(\v{t}^{(1)})$ and similarly for $\v{t}^{(3)}$ and $\v{t}^{(4)}$. Nevertheless, $\v t^{(1)}$ is incomparable to $\v t^{(4)}$; similarly $\v t^{(2)}$ belongs to the incomparable region of $\v t^{(3)}$ [see Fig.~\hyperref[fig-lattice-diagram]{\ref{fig-lattice-diagram}b}].
From this we find the non-inclusions $\mathcal{T}_{+}(\v t^{(1)}) \not\subset \mathcal{T}_{+}(\v t^{(4)})$ and $\mathcal{T}_{+}(\v t^{(4)}) \not\subset \mathcal{T}_{+}(\v t^{(1)})$, similarly $\mathcal{T}_{+}(\v t^{(2)}) \not\subset \mathcal{T}_{+}(\v t^{(3)})$ and $\mathcal{T}_{+}(\v t^{(3)}) \not\subset \mathcal{T}_{+}(\v t^{(2)})$. Naively, one may be led to a conclusion that the incomparable region can be characterised by the future cones of $\v t^{(1)}$ and $\v t^{(4)}$ alone. However, any convex combination $\lambda \v{t}^{(2)} + (1-\lambda)\v {t}^{(3)} \equiv \v t^{(2,3)}_\lambda$ results in an incomparable vector $\v{t}^{(2,3)}_\lambda \in \mathcal{T}_\emptyset(\v{t}^{(i)})$ for $i = 1,2,3,4$ and $0<\lambda<1$
[see Fig.~\hyperref[fig-lattice-diagram]{\ref{fig-lattice-diagram}c}], and hence, in a ``new'' fragment of the incomparable region. In order to account for the entire incomparable region one must take the union of all future cones of $\v t^{(2,3)}_\lambda$ for $\lambda \in [0,1]$. This corresponds to the convex hull of the future cones\footnote{Geometrically, the mixture of $\v t^{(2)}$ and $\v t^{(3)}$ corresponds to the edge that connects these two points.} $\mathcal{T}_{+}(\v t^{(2)})$ and $\mathcal{T}_{+}(\v t^{(3)})$. Furthermore, the construction is limited only to convex combinations of futures for consecutive tangent vectors since mixtures of two non-successive ones, for instance $\v t^{(1)}$ and $\v t^{(4)}$, do not give any additions to the incomparable region. Such combinations belong to the past cone of $\v p$ as every point of $\lambda \v{t}^{(1)} + (1-\lambda)\v{t}^{(4)}$ would majorise $\v{p}$.


\subsection{Links to other resource theories}
\label{SubSec_LinkstootherQuantumTheorioes}

\begin{figure}[t]
    \centering
    \includegraphics{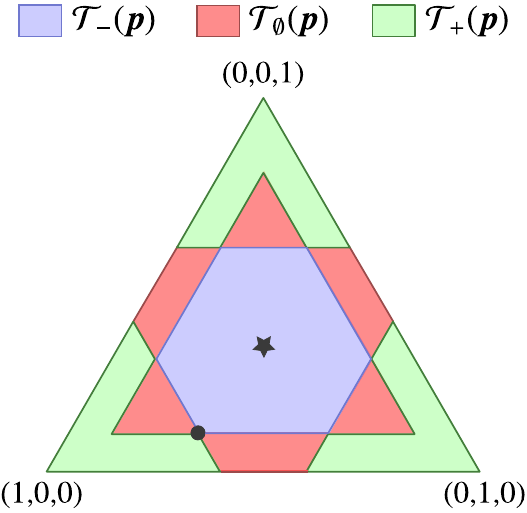}
    \caption{\label{fig-entanglementcone} \emph{Entanglement cone in the simplex of the Schmidt coefficients of a $3\times 3$ system}. Conversely to thermodynamics, the past of entanglement transformations is the thermodynamic future and vice-versa. The black dot $\bullet$ indicates the Schmidt vector of the initial state \mbox{$\v p = (0.7,0.2,0.1)$}, whereas the black star $\bigstar$ represents the maximally entangled state \mbox{$\v \eta = (1/3,1/3,1/3)$}.}
\end{figure}

The two well-known examples of majorisation-based resource theories, where our results are also applicable, include the resource theories of entanglement and coherence. These are defined via the appropriate sets of free operations and free states: local operations and classical communication (LOCC) and separable states in entanglement theory~\cite{Horodecki2009}, and incoherent operations (IO) and incoherent states in coherence theory~\cite{Plenio2014}. Within each of these theories, there exists a representation of quantum states via probability distributions that is relevant for formulating state interconversion conditions under free operations. In entanglement theory, a pure bipartite state \mbox{$\rho = \ketbra{\Psi}{\Psi}$} can be written in terms of the Schmidt decomposition given by $\ket{\Psi} = \sum_i a_i \ket{\psi_i, \psi'_i}$, and represented by a probability vector $\v{p}$ with $p_i = |a_i|^2$. Then, Nielsen’s theorem~\cite{nielsen1999conditions} states that an initial state $\v p$ can be transformed under LOCC into a target state $\v q$ if and only if $\v p \prec \v q$. Similarly, in the resource theory of coherence with respect to a fixed basis $\{ \ket{i}\}$, one can represent a pure state $\rho = \ketbra{\psi}{\psi}$ by a probability vector $\v p$ with $p_i = |\langle i | \psi \rangle|^2$. Then, a given initial state $\v p$ can be transformed into $\v q$ via incoherent operations if and only if $ \v p \prec \v q$~\cite{Du2015}.

Therefore, we observe that the partial order emerging in the two cases is precisely the opposite to the thermodynamic order in the infinite temperature limit (For more details see Ref.~\cite{Kollas_master}). Consequently, the thermodynamic past and future become the future and past for entanglement and coherence, while the incomparable region remains unchanged~(see Fig.~\ref{fig-entanglementcone}). \rd{Note that, for entanglement and coherence, sharp states $\v s$ are in the future cone of any given state, while for thermodynamics (at $\beta = 0$), they are in the past. The flat distribution $\v \eta$ is in the past of any state in entanglement and coherence theories, whereas in thermodynamics it is in the future.}

\rd{One can make a general remark concerning resource monotones, applying to the entanglement, coherence and thermodynamic scenarios alike. Consider an entangled state \mbox{$\ket\psi \in \mathcal{H}_N\otimes \mathcal{H}_N$} with the associated Schmidt coefficients $\v{p}$ and concurrence $\mathcal{C}(\ket{\psi})$ as an example of a resource monotone \cite{HW97}. If another state $\ket{\phi}$ with Schmidt coefficients $\v{q}$ is in the future cone of $\ket{\psi}$, $\v{q}\in\mathcal{T}_+(\v{p})$, then $\mathcal{C}(\ket{\phi}) \leq \mathcal{C}(\ket{\psi})$. Otherwise, if it lies in its past cone, $\v{q}\in\mathcal{T}_-(\v{p})$ we know that $\mathcal{C}(\ket{\phi}) \geq \mathcal{C}(\ket{\psi})$. However, if the two states are incomparable,  $\v{q}\in\mathcal{T}_\emptyset(\v{p})$, nothing can be said about the relation between both concurrences.}

\begin{figure*}
    \centering
    \includegraphics{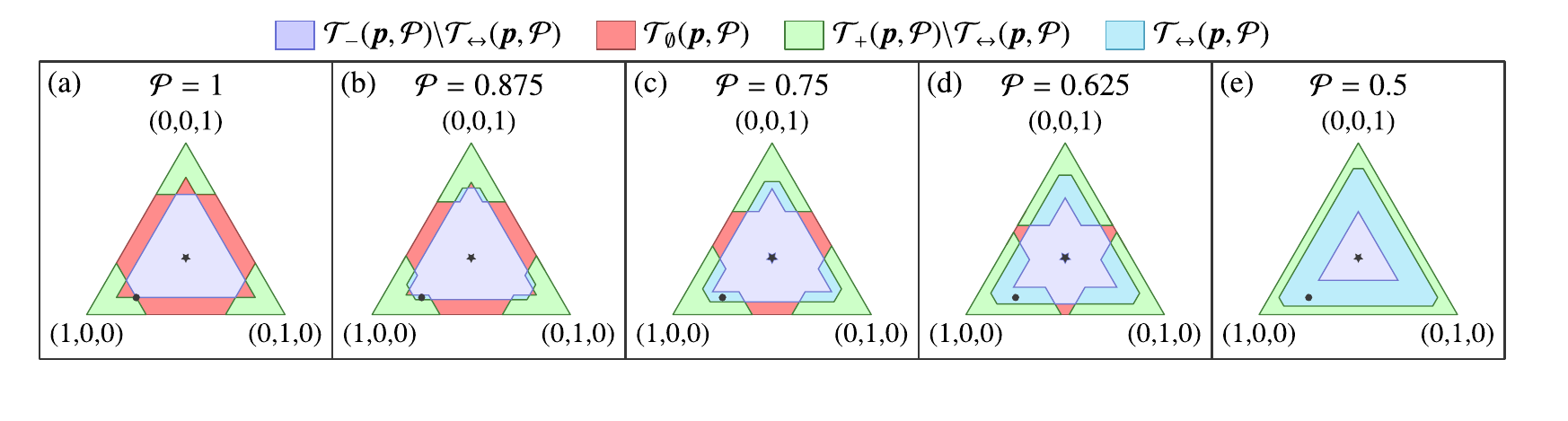}
    \caption{\emph{Probabilistic majorisation cones for $d=3$}. \rd{For a three-level system with a state given by $\v{p} = (0.7,0.2,0.1)$ represented by a black dot~$\bullet$ and a maximally entangled state $\v \eta = (1/3,1/3,1/3)$ represented by a black star $\bigstar$, we plot its probabilistic majorisation cone for probabilities of transformation $\mathcal{P}$ decreasing from  $1$ to $0.5$ with $0.125$ steps (a-e), respectively. Observe that, for $\mathcal{P} = 1$, we recover the structure of the standard majorisation cones, while as $\mathcal{P}$ decreases the interconvertible region $\mathcal{T}_{\leftrightarrow}(\v{p},\mathcal{P})$ expands and the incomparable region $\mathcal{T}_{\emptyset}(\v{p},\mathcal{P})$ shrinks, disappearing altogether between panel (d)~and~(e)}.}
    \label{fig-probabilistic-cones-examples} 
\end{figure*}

\subsection{Probabilistic majorisation cones} \label{sec:prob_maj_con}

\rd{Finally, it should be observed that the notion of majorisation cones, as presented until now, deals with deterministic transformations. However, this approach can be extended to probabilistic transformations using Vidal's criterion for entanglement~\cite{vidal1999entanglement,Vidal2000} and coherence transformations~\cite{ZhuEtAl2017coherence} under LOCC and IO, respectively. In the case of probabilistic transformations of bipartite entangled states under LOCC, this is captured by the following theorem.} 
\rd{

\begin{restatable}[Theorem~1 of Ref.~\cite{vidal1999entanglement}]{thm}{vidalentanglement}\label{thm_vidal}
    Consider two bipartite pure states $\ket{\psi}$ and $\ket{\phi}$, whose Schmidt decompositions are described by probability vectors $\v p$ and $\v q$, respectively. The maximal transformation probability from $\ket{\psi}$ to $\ket{\phi}$ under LOCC is given by
    \begin{equation} \label{eq:probal_crit_IO}
        \mathcal{P}(\v{p}, \v{q}) = \min_{1\leq k \leq d} \frac{\sum_{j = k}^d p^{\downarrow}_j}{\sum_{j = k}^d q
        ^{\downarrow}_j}.
    \end{equation}
\end{restatable}

In Appendix~\ref{app:probal_deriv} we discuss the extension of majorisation cones to probabilistic ones, denoted as $\mathcal{T}_i(\v p; \mathcal{P})$, with \mbox{$i \in \{-, \emptyset, + \}$} and $\mathcal{P}$ being the minimal probability of transformation. Here we will limit ourselves to a brief qualitative discussion
about the behaviour of the probabilistic majorisation cones as the transformation probability changes from $\mathcal{P}(\cdot,\cdot) = 1$ to $\mathcal{P}(\cdot,\cdot) < 1$ (see Fig.~\ref{fig-probabilistic-cones-examples}). Note that the only common points of the future and past for $\mathcal{P} = 1$ are the current state of the system $\v p$ and its permutations. Conversely, for $\mathcal{P} < 1$ this is not the only case; consequently, we may define the \textit{interconvertible region} of $\v p$ at the probability level $\mathcal{P}$ as the intersection between the probabilistic past and probabilistic future, $\mathcal{T}_{\leftrightarrow}(\v p, \mathcal{P}) \equiv \mathcal{T}_+(\v p, \mathcal{P}) \cap \mathcal{T}_-(\v p, \mathcal{P})$. This region is non-empty for every $\mathcal{P} < 1$. It is easily shown that the future and the past cones grow as the probability of transformation decreases, $\mathcal{T}_+(\v p, \mathcal{P}') \subset \mathcal{T}_+(\v p, \mathcal{P}) $ and $\mathcal{T}_-(\v p, \mathcal{P}') \subset \mathcal{T}_-(\v p, \mathcal{P})$ for $\mathcal{P}' > \mathcal{P}$. Therefore, the only region that decreases together with $\mathcal{P}$ is the incomparable region, $\mathcal{T}_\emptyset(\v p, \mathcal{P}') \supset \mathcal{T}_\emptyset(\v p, \mathcal{P})$. Interestingly, for every state~$\v p$, we observe that there is a critical value $\mathcal{P}^*$, at which no two states are incomparable, i.e., $\mathcal{T}_\emptyset(\v p, \mathcal{P}) = \emptyset$~[see Figs.~\hyperref[fig-probabilistic-cones-examples]{\ref{fig-probabilistic-cones-examples}a,~e}]. Analogous results hold in the context of coherence, as Theorem~\ref{thm_vidal} has its counterpart when considering pure state transformations under IO operations~\cite{ZhuEtAl2017coherence,Cunden2021}.
}

\rd{Finally, it is worth mentioning, that a criterion similar to the Vidal's criterion was established for probabilistic transformation in the context of thermal operations~\cite{AOP16}. In this case, the construction of the probabilistic cones for majorisation generalises directly to the thermomajorisation by using the construction of thermomajorisation cones which we will introduce in the next section.}


\section{Thermal cones} 
\label{Sec:Thermal-cones}

Let us turn our attention to a more general scenario, assuming that the temperature is finite, $\beta> 0$. In this case, the reachability of energy-incoherent states under Gibbs-preserving matrices can be studied by introducing the notion of \emph{thermal cones}, defined as follows:

\begin{defn}[Thermal cones]\label{def_ThermalCones}
The set of states that an energy-incoherent state $\v p$ can be mapped into by Gibbs-preserving matrices is called the \textbf{future thermal cone} $\T^{\, \beta}_+(\v{p})$. Similarly, the set of states that can be mapped to $\v p$ by Gibbs-preserving matrices is called the \textbf{past thermal cone} $\T^{\, \beta}_-(\v{p})$. Finally, the set of states that are neither in the past nor in the future of $\v p$ is called the \textbf{incomparable thermal region} $\T^{\, \beta}_{\emptyset}(\v{p})$.
\end{defn}

Despite apparent similarities, the case of $\beta > 0$ turns out to be significantly harder than $\beta=0$. Difficulties stem mostly from a simple fact demonstrated in Ref.~\cite{Korzekwa2017} -- even though thermomajorisation forms a lattice in each $\beta$-order, it does not provide a common lattice for the entire probability simplex. Thus, before extending Lemma~\ref{lem_incomparablecone} and Theorem~\ref{thm_pastcone} to the thermal setting (proofs of which rely heavily on the existence of a join), we will introduce an \emph{embedding lattice} -- a structure which encompasses thermomajorisation order as its subset -- and we will demonstrate operations shifting to and from the newly introduced picture.


\subsection{Embedding lattice} 
\label{sec_embeddingLattice}

To define a $\beta$-dependent embedding $\mathfrak{M}$ of the simplex $\Delta_d$ into a subspace $\Delta^{\mathfrak{M}}_d\subset\Delta_{2^d - 1}$, illustrated by a graphical example in Fig.~\ref{fig-embeddingscheme}, we first introduce the vector $\v \Gamma^{\mathfrak{M}}$ with entries given by all possible partial sums of the Gibbs distribution,
\begin{equation}
    \v \Gamma^{\mathfrak{M}} = \left\{\sum_{i\in I}\gamma_i:\,I\in2^{\{1,\hdots,d\}}\right\},
\end{equation}
with $2^{\{1,\hdots,d\}}$ denoting the power set of $d$ indices. Moreover, we enforce that it is ordered non-decreasingly, i.e., for \mbox{$i > j$} we have \mbox{$\Gamma^{\mathfrak{M}}_i \geq \Gamma^{\mathfrak{M}}_j$}. Then, the embedded probability vector, \mbox{$\v{p}^{\mathfrak{M}}:=\mathfrak{M}(\v p)$}, is defined by 
\begin{align}
    p^{\mathfrak{M}}_i & = f^{\, \beta}_{\v p}\left(\Gamma^{\mathfrak{M}}_i\right) -f^{\, \beta}_{\v p}\left(\Gamma^{\mathfrak{M}}_{i-1}\right) ,
\end{align}
where $f^{\, \beta}_{\v p}$ is the thermomajorisation curve of $\v{p}$.

Within this embedding, the thermomajorisation indeed proves to be almost-standard majorisation relation between the embedded distributions,  
\begin{equation} \label{eq_embedSpace_majorisationDefinition}
    \v p \succ_\beta \v q ~\Longleftrightarrow~ \forall j:~\sum_{i=1}^j p^{\mathfrak{M}}_i \geq \sum_{i=1}^j q^{\mathfrak{M}}_i ~\overset{\text{def}}{\Longleftrightarrow}~\v {p}^{\mathfrak{M}}\succ_{\mathfrak{M}}\v{q}^{\mathfrak{M}},
\end{equation}
where the last symbol $\succ_{\mathfrak{M}}$ denotes the majorisation variant related to the embedding lattice. Finally, we note that the only deviation from the standard majorisation lies in the convexity condition in the embedding space, which is imposed not on the probabilities $p^{\mathfrak{M}}_i$ themselves, but on their rescaled versions, 

\begin{equation} \label{eq_embedSpace_orderingRule}
    i \leq j \Rightarrow \frac{p^{\mathfrak{M}}_i}{\gamma^{\mathfrak{M}}_i} \geq \frac{p^{\mathfrak{M}}_j}{\gamma^{\mathfrak{M}}_j} ,
\end{equation}
with scaling factors directly related to the embedded majorisation curve of the Gibbs state $\v{\gamma}^{\mathfrak{M}}$. This ordering should be compared (but not confused) with the $\beta$-ordering introduced in Eq.~\eqref{eq_beta-ordering}, pointing to a relation with thermomajorisation which we will use to show the lattice structure of the introduced space.

The projection $\mathfrak{P}_{\v \pi}$ of an arbitrary probability vector \mbox{$\v{q}\in\Delta_{2^d-1}$} satisfying Eq.~\eqref{eq_embedSpace_orderingRule} onto a selected $\beta$-order $\v \pi$ in the original space can be defined descriptively as taking only those elbows of the embedded majorisation curve that match the values of cumulative Gibbs distribution for the selected permutation. Formally, the projected vector, \mbox{$\v q^{\mathfrak{P}}_{\v{\pi}}:=\left(\mathfrak{P}_{\v \pi}\left(\v q\right)\right)^{\, \beta}$}, is entry-wise defined by 
\begin{equation} \label{eq_embedding_porj_constr}
   \left(q^\mathfrak{P}_{\v{\pi}}\right)_i = \sum^{k(i)}_{j = k(i-1)} q_j ,
\end{equation}
with the indices $k(i)$ defined by the requirement that \mbox{$\Gamma^\mathfrak{M}_{k(i)} = \sum_{j=1}^i \gamma_{\v \pi^{-1}(j)}$}. 

In particular, it is worth noting two properties of the embedding $\mathfrak{M}$ and projections $\mathfrak{P}_{\v \pi}$. First, given a vector $\v p \in\Delta_d$ with a $\beta$-order $\v \pi_{\v p}$, we find that by construction $\mathfrak{P}_{\v \pi_{\v p}}\left(\mathfrak{M}\left(\v p\right)\right) = \v p$, which follows directly from Eq.~\eqref{eq_embedding_porj_constr}. On the other hand, for $\v \pi \neq \v \pi_{\v p}$  we find that $\v p \succ_\beta \mathfrak{P}_{\v \pi}\left(\mathfrak{M}\left(\v p\right)\right)$. The statement is easily shown by observing that the Lorenz curve of $\mathfrak{P}_{\v \pi}\left(\mathfrak{M}\left(\v p\right)\right)$ connects by line segments $d+1$ points of the Lorenz curve corresponding $f^\beta_{\v p}\left(\sum_i \left(\mathfrak{M}(p)^\mathfrak{P}_{\v{\pi}}\right)_i\right)$ and therefore majorisation is resolved by linear approximation of a convex function, $f[(1-t)x + t y] \geq (1-t)f(x) + tf(y)$ for any $x, y$. The second property is concerned with $\v q \in\Delta_{2^d-1}$ satisfying Eq.~\eqref{eq_embedSpace_orderingRule} and can be summarised as the fact that projecting and re-embedding the vector will always give the object majorised by the original vector: $\v q\succ_\mathfrak{M} \mathfrak{M}(\mathfrak{P}_{\v \pi}(\v q))$. It follows similarly to the prior majorisation by the argument of linear approximation of a convex function.

It is necessary to stress that the introduced embedding structure is distinct from the one used in the usual method of reducing thermomajorisation to majorisation~\cite{horodecki2013fundamental}. Most importantly, the standard approach requires going to the limit of infinite embedding dimension in order to recover thermomajorisation for arbitrary $\beta$ as a special case of standard majorisation. In our proposition, a thermomajorisation curve in dimension $d$ is embedded within a $2^d-1$-dimensional space. The main difference lies in the non-constant widths of the segments of the Lorenz curve and the fact that once the embedded vector $\v{p}^{\mathfrak{M}}$ is constructed, its entries should not be subject to reordering. Finally, we present the argument proving that the embedding together with embedded majorisation indeed provide a lattice structure.

\begin{figure*}
    \centering
    \includegraphics{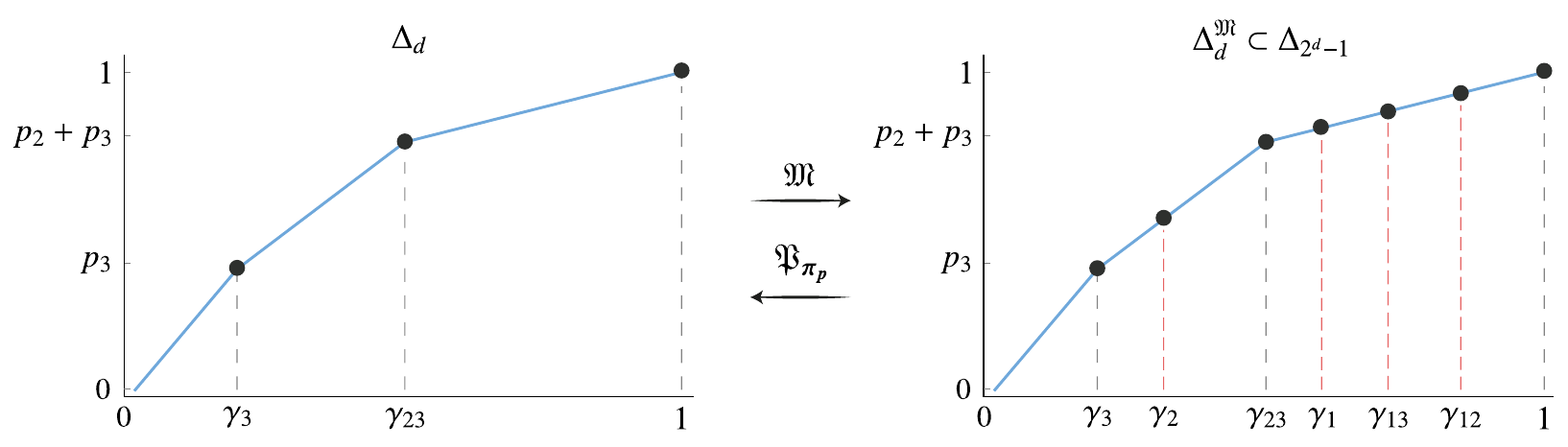}
    \caption{\label{fig-embeddingscheme} \emph{Embedded majorisation curve}. Embedding $\mathfrak{M}:\Delta_d\rightarrow\Delta_{2^d-1}$ of $d$-dimensional probability distribution $\v{p}$ (here $d=3$) into a $2^d-1 = 7$-dimensional space is most easily understood by noting that each thermomajorisation curve (left) has elbows corresponding to a subset of $\v{\Gamma}^{\mathfrak{M}}$ on the horizontal axis. The embedding includes all entries of $\v{\Gamma}^{\mathfrak{M}}$ by subdividing the Lorenz curve into $2^d - 1$ fragments. Conversely, the projection $\mathfrak{P}_{\v \pi_{\v p}}$ corresponds to selecting only a subset of of elbows that correspond to a selected order, in this case the original order $\v{\pi}_{\v p}$. In the $x$-axes we used the shorthand notation \mbox{$\gamma_{ij} = \gamma_i + \gamma_j$}. }
\end{figure*}

\begin{cor}
    The subset of the probability simplex $\Delta_{2^d - 1}$ satisfying Eq.~\eqref{eq_embedSpace_orderingRule} and subject to the embedded majorisation $\succ_\mathfrak{M}$ defined in Eq.~\eqref{eq_embedSpace_majorisationDefinition} forms a lattice.
\end{cor}

\begin{proof}
The embedded majorisation $\succ_\mathfrak{M}$ may be reinterpretted as thermomajorisation defined for a specific Gibbs state $\v{\gamma}^\mathfrak{M}$ and restricted to a particular Weyl chamber of the probability space $\Delta_{2^d - 1}$ by comparing Eq.~\eqref{eq_embedSpace_orderingRule} with an analogous sorting rule from Eq.~\eqref{eq_beta-ordering}. This direct isomorphism between thermomajorisation order $\succ_\beta$ restricted to a single Weyl chamber, known to provide a lattice structure, and the embedded majorisation $\succ_\mathfrak{M}$ proves the statement.
\end{proof}


\subsection{Geometry of thermal cones}

As already mentioned, for finite temperatures the rules underlying state transformations are no longer captured by a majorisation relation, but rather by its thermodynamic equivalent known as thermomajorisation~\cite{Rusch,horodecki2013fundamental}. As a result, Birkhoff’s theorem cannot be employed anymore, and the characterisation of the future thermal cone is no longer given by Theorem~\ref{thm_futureinfinite}. However, the set of Gibbs-preserving matrices still forms a convex set~\cite{mazurek2018decomposability,mazurek2019thermal}, and the extreme points of the future thermal cone can be constructed by employing the following lemma:
\begin{lem}[Lemma 12 of Ref.~\cite{Lostaglio2018elementarythermal}]
    \label{lem_extreme}
	Given $\v{p}$, consider the following distributions $\v{p}^{\, \v \pi}\in \T^{\, \beta}_{+}(\v{p})$ constructed for each permutation $\v \pi\in \S_d$. For $i\in\left\{1,\dots,d\right\}$:
	\begin{enumerate}
		\item Let $x_i^{\v \pi}=\sum_{j=0}^{i} e^{-\beta E_{\v \pi^{-1}\left(j\right)}}$ and $y_i^{\v \pi}=f^{\, \beta}_{\v{p}}\left(x_i^{\v \pi}\right)$. 
		\item Define $p^{\v \pi}_i:=y^{\v \pi}_{\v \pi(i)} - y^{\v \pi}_{\v \pi(i)-1}$, with $y_{0}:=0$.
	\end{enumerate}
	Then, all extreme points of $\T^{\, \beta}_{+}(\v{p})$ have the form $\v{p}^{\v \pi}$ for some ${\v \pi}$. In particular, this implies that $\T^{\, \beta}_{+}(\v{p})$ has at most $d!$ extremal points.
\end{lem}

The above lemma allows one to characterise the future thermal cone of $\v p$ by constructing states $\v p^{\, \v \pi}$ for each $\v \pi \in \mathcal S_d$, and taking their convex hull. It is worth mentioning that $\T^{\, \beta}_{+}(\v p)$ can also be constructed by finding the whole set of extremal Gibbs-preserving matrices~\cite{Lostaglio2018elementarythermal, mazurek2018decomposability}. This follows the same spirit as in Sec.~\ref{Sec:Majorisation-cones}, where the majorisation cone was characterised by employing Theorem~\ref{thm_HLP}. However, this is a harder problem to solve, and so the extremal Gibbs-preserving matrices were characterised only for $d \leq 3$~\cite{mazurek2018decomposability,mazurek2019thermal}. Here, we provide the construction of $\T^{\, \beta}_{+}(\v p)$ as a simple corollary of Lemma~\ref{lem_extreme}.
\begin{figure*}
    \centering
    \includegraphics{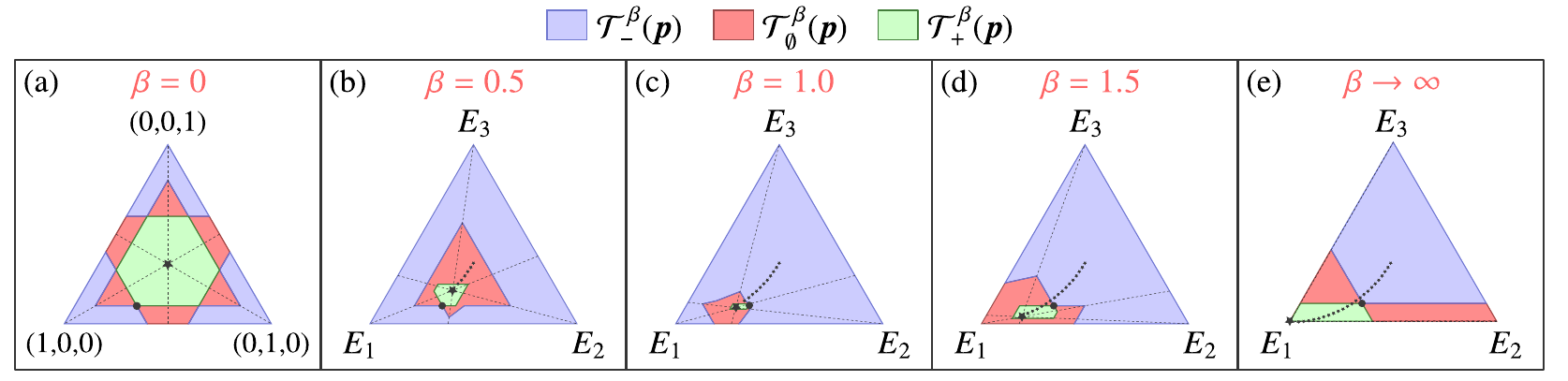}
    \caption{\emph{Thermal cones for $d=3$}. For a three-level system with population given by $\v p = (0.4, 0.36, 0.24)$, represented by a black dot $\bullet$, and energy spectrum $E_1 = 0, E_2 = 1$ and $E_3 = 2$, we plot its thermal cone for (a) $\beta = 0$, (b) $\beta = 0.5$, (c) $\beta = 1.0$ and (d) $\beta \to \infty$. By increasing $\beta$, the thermal state (black star $\bigstar$) tends toward the ground state $E_1$, and the past thermal cone becomes convex.} 
    \label{fig-thermal-cones-examples} 
\end{figure*}

\begin{cor}[Future thermal cone]
\label{thm_futurefinite}
The future thermal cone of a $d$-dimensional energy-incoherent state $\v p$ is given by
\begin{equation}
\T^{\,\beta}_+(\v p) = \operatorname{conv}[\{\v p^{\v \pi}, \v \pi\in\S_d\}] \, .
\end{equation}
\end{cor}

Furthermore, one can use the embedding lattice to provide an alternative formulation for the future thermal cone:
\begin{obs}
    Since $\v{p}^{\v{\pi}} = \mathfrak{P}_{\v{\pi}}(\mathfrak{M}(\v{p}))$, the future thermal cone of an energy-incoherent state $\v p$ can be expressed in terms of all possible projections from the related embedding,
    \begin{equation}
        \mathcal{T}^\beta_+(\v{p}) = \operatorname{conv}\left[\left\{\mathfrak{P}_{\v{\pi}}(\mathfrak{M}(\v{p}))\,, \v \pi\in\mathcal{S}_d\right\}\right].
    \end{equation}
\end{obs}
Our main technical contribution is captured by the following lemma that generalises Lemma~\ref{lem_incomparablecone} and provides the construction of the incomparable thermal region for finite temperatures. Its proof is based on the concept of embedding lattice that we introduced in Sec.~\ref{sec_embeddingLattice} and can be found in Appendix~\ref{app:finite_temp_derivs}.

\begin{lem}[Incomparable thermal region] 
\label{lemma_incomparablefiniteT}
Given an energy-incoherent state $\v p$ and a thermal state $\v \gamma$, consider distributions $\v t^{(n,\v \pi)}$ in their $\beta$-ordered form, constructed for each permutation $\v \pi \in \S_{d}$,
\begin{equation}
\label{eq_thermaltangentvectors}
\left(\v{t}^{(n,\v \pi)}\right)^{\, \beta} = \left(
    t^{(n,\v \pi)}_{\v \pi(1)}, 
    p^{\, \beta}_n\frac{\gamma_{\v \pi (2)}}{\gamma^{\, \beta}_n}, ..., 
    p^{\, \beta}_n\frac{\gamma_{\v \pi (d-1)}}{\gamma^{\, \beta}_n},
    t^{(n,\v \pi)}_{\v \pi(d)}\right) ,  
\end{equation}
with
\begin{subequations}
\begin{align}
    t^{(n,\v \pi)}_{\v \pi(1)} & = \sum_{i=1}^n p^{\, \beta}_i-\frac{p^{\, \beta}_n}{\gamma^{\,\beta}_n}\left(\sum_{i=1}^n \gamma^{\,\beta}_i-\gamma_{\v \pi(1)} \right), \\ 
    t^{(n,\v \pi)}_{\v \pi(d)} & = 1- t^{(n,\v \pi)}_{\v \pi(1)}-\frac{p^{\, \beta}_n}{\gamma^{\,\beta}_n}\sum_{i=2}^{d-1} \gamma_{\v \pi(i)} .
\end{align}
\end{subequations}
Defining the set
\begin{equation}
\label{eq_incomparableconebeta}
    \mathbb{T}^{\,\beta} = \bigcup_{\v \pi\in\mathcal{S}_d}\bigcup_{i=1}^{n-1}\operatorname{conv}\left[\mathcal{T}^{\,\beta}_+\left(\v{t}^{(i,\v \pi)}\right)\cup\mathcal{T}^{\,\beta}_+\left(\v{t}^{(i+1,\v \pi)}\right)\right], 
\end{equation}
the incomparable region of $\v{p}$ is given by 
\begin{align}
\label{eq_incomparableconeFT}
\T^{\, \beta}_{\emptyset}(\v{p}) = \left[\operatorname{int}\left(\mathbb{T}^{\, \beta}\right)\backslash\T^{\, \beta}_+(\v p)\right]\cap \v{\Delta}_d .
\end{align}
\end{lem}

Analogously to the infinite temperature case, Lemma~\ref{lemma_incomparablefiniteT} allows us to obtain the past thermal cone of $\v p$.
\begin{thm}[Past thermal cone]
\label{thm_pasthermalconefinite}
The past thermal cone of $\v p$ is given by \begin{equation}
\T^{\, \beta}_{-}(\v p)=\v{\Delta}_d \backslash \operatorname{int} \left(\mathbb{T}^{\, \beta}\right) \, .
\end{equation}
\end{thm}
\begin{proof}
Following the same reasoning as in the proof of Theorem~\ref{thm_pastcone}, we only need to use the fact that \mbox{$\operatorname{int}(\mathbb{T}^\beta) = \mathcal{T}^\beta_+(p) \cup \mathcal{T}^\beta_\emptyset(p)$}.   
\end{proof}

In Figs.~\ref{fig-thermal-cones-examples} and \ref{fig-thermal-cones-3dexamples}, we illustrate Lemma~\ref{lemma_incomparablefiniteT} and Theorem~\ref{thm_pasthermalconefinite} for a three-level system and a four-level system in different temperature regimes. We also provide a \texttt{Mathematica} code~\cite{adeoliveirajunior2022} that constructs the set of extreme points of the future and past thermal cones for arbitrary dimensions.

\begin{figure}[t]
    \centering
    \includegraphics{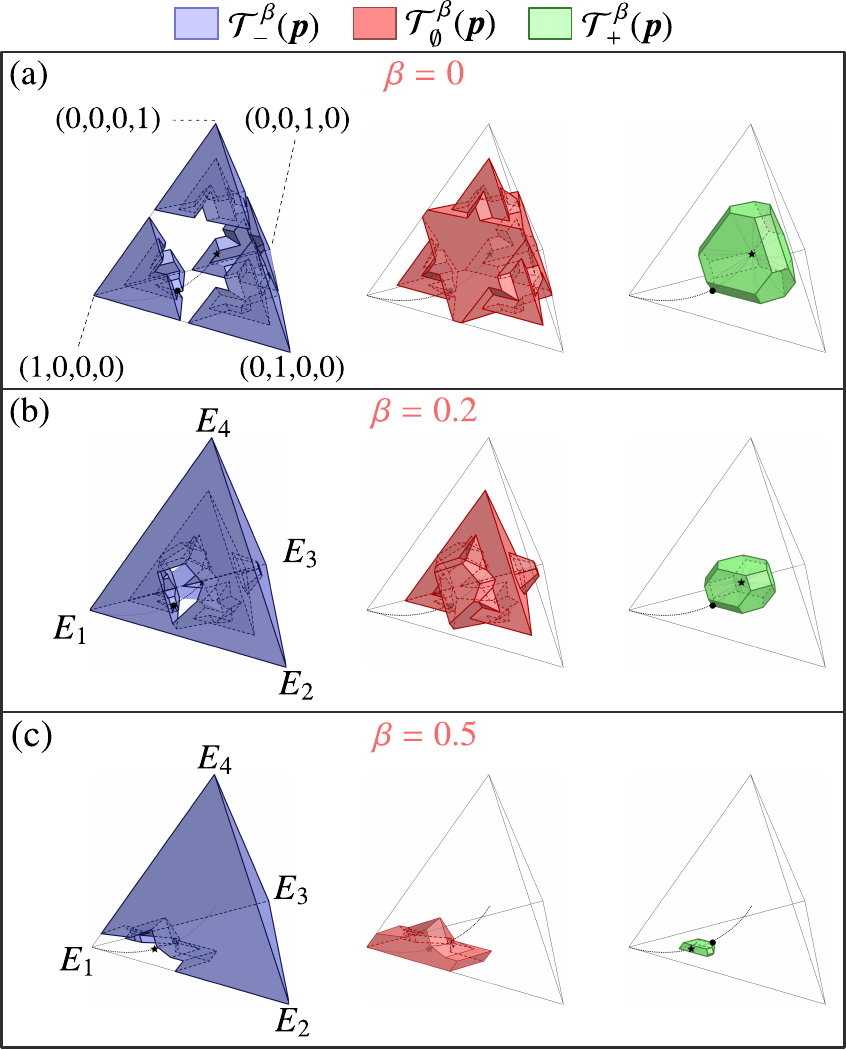}
    \caption{\emph{Thermal cones for $d=4$}. For a four-level system with population given by $\v p = (25,13,7,3)/48$, represented by a black dot $\bullet$, and energy spectrum $E_1 = 0, E_2 = 1, E_3 = 2$ and $E_4 =3$, we plot its thermal cone for (a) $\beta = 0$, (b) $\beta = 0.2$ and (c) $\beta = 1.0$. The thermal state depicted by a black star $\bigstar$, and its trajectory (dashed line) is also shown. By increasing the inverse temperature $\beta$, the geometry of each region drastically changes.}
    \label{fig-thermal-cones-3dexamples} 
\end{figure}

As before, the past thermal cone forms a convex polytope only when restricted to a single Weyl chamber, now defined as a set of probability vectors with common $\beta$-order. The extreme points of the past thermal cone correspond to tangent vectors $\v t^{(n,\v \pi)}$ or by their projection onto the boundary of the probability simplex. The exceptional points in comparison with the infinite-temperature case may appear when considering extreme points of \mbox{$\bigcup_{i=1}^{d-1} \operatorname{conv}[(\mathcal{T}^\beta_+(\v{t}^{(i,\v \pi)})\cup\mathcal{T}^\beta_+(\v{t}^{(i+1,\v \pi)})]$} for a given chamber $\v \pi$. Vertices arising in this way may not correspond to any tangent vector (see Fig.~\ref{fig-past_chamberextreme}). Moreover, the $\v t^{(n,\v \pi)}$ vectors are also responsible for the convexity of the past thermal cone, and the following observation illustrates this:
\begin{cor}
Approaching the limit of $\beta \to \infty$, the past thermal cone becomes convex.
\end{cor}
\begin{proof}
By dividing the probability simplex into equal chambers with the thermal state in the barycenter, the past thermal cone is the union of $d!$ convex pieces. As $\beta \to \infty$, the thermal state collapses to the ground state. There is only one chamber in this limit, and therefore, the past thermal cone is a single convex piece [see Fig.\hyperref[fig-thermal-cones-examples]{\ref{fig-thermal-cones-examples}~d} for an example in the particular case of $d=3$].
\end{proof}

\begin{figure*}
    \centering
    \includegraphics{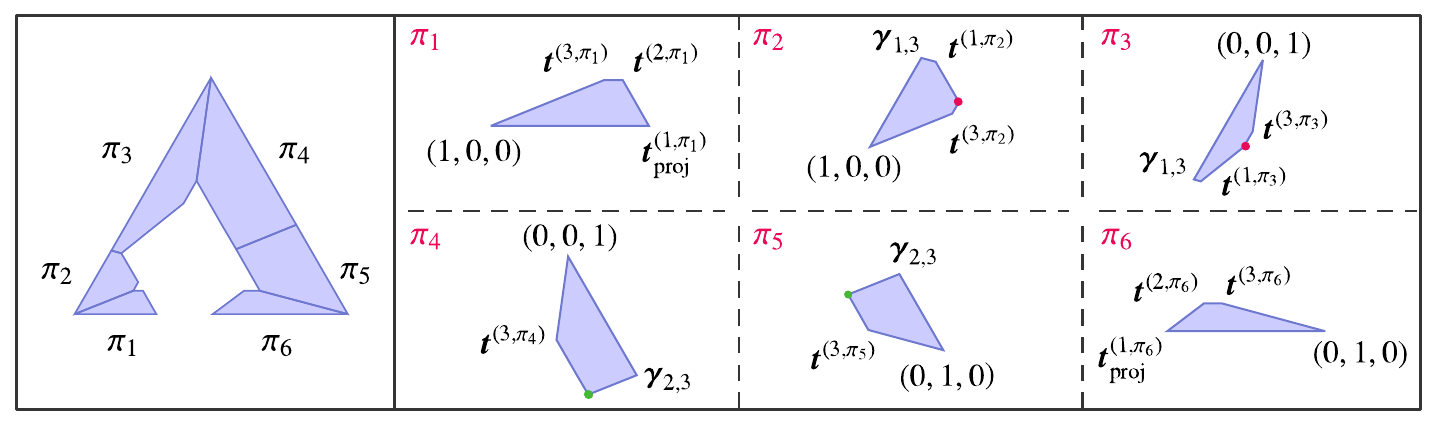}
    \caption{\label{fig-past_chamberextreme} \emph{Extreme points of $\mathcal{T}^{\,\beta}_-(\v p)$}. The past thermal cone for a three-level system with population give by $\v p = (0.7,0.2,0.1)$, energy spectrum $E_1 = 0, E_2 = 2$ and $E_3 = 3$, and the inverse temperature $\beta = 0.5$. At each chamber $\v \pi$, the non-trivial extreme points of the past are given by $\v t^{(n, \v \pi)}$ (for $n \in \{1, 2, 3\}$ and $\v \pi \in \mathcal{S}_d$) and by extreme points of \mbox{$\bigcup_{i=1}^{d-1} \operatorname{conv}[(\mathcal{T}^{\, \beta}_+(\v{t}^{(i,\v \pi)})\cup\mathcal{T}^{\, \beta}_+(\v{t}^{(i+1,\v \pi)})]$} (red dots). The remaining ones are sharp states, points in the boundary between chambers (green dots) and those which are on the edge of the probability simplex, i.e., $(\v \gamma_{i,j})_{k} = (\gamma_i \delta_{ik}+\gamma_j \delta_{ik})/(\gamma_i + \gamma_j)$, for $i\neq j$, and $k \in \{1, 2,3\}$.}
    \end{figure*}

The intuition behind the above observation can be understood by studying the behaviour of a three-level system and the tangent vector $\v t^{(3,(132))}$. By a decreasing temperature, this extremal point tends towards the edge of the simplex and reaches the edge at $\beta \to \infty$ [see Figs.~\hyperref[fig-thermal-cones-examples]{\ref{fig-thermal-cones-examples}c,~d}].

To wrap up the considerations of the geometry of thermal cones, let us go back, once again, to the analogy between the thermal cones and special relativity. Consider that given a specific division of space-time into future, past and space-like regions, one is able to recover the specific event generating it, which we may refer to as Present. Concisely -- there is a one-to-one relation between events and the divisions of space-time they generate. The situation is exactly reflected for thermal cones with $\beta > 0$ -- given a specific arrangement of incomparable region and future and past thermal cones, one can exactly recover the current state of the system [see the black dot~$\bullet$ in Fig.~\hyperref[fig-thermal-cones-examples]{\ref{fig-thermal-cones-examples}b-d}]. It is in stark contrast with the majorisation cones for $\beta = 0$, where every division into past, future and incomparable possesses $d!$-fold symmetry [see Fig.~\hyperref[fig-thermal-cones-examples]{\ref{fig-thermal-cones-examples}a}] and hence, the present state of the system cannot be recovered solely on its basis unless provided with additional information like the permutation which sorts the probabilities in non-decreasing order.



\section{Volume of thermal cones} \label{Sec_volume_thermal_cones}

In the previous sections, we characterised and discussed the behaviour of thermal cones by introducing explicit constructions of the past, the future and the incomparable region. It is then natural to ask what is the role played by their volumes in quantifying the resourcefulness of different states~\cite{Cunden_2020,Cunden2021}. Thus, for a $d$-dimensional probability vector $\v p$
we define the relative volumes of its thermal cones as
\begin{equation}
    \mathcal{V}^{\, \beta}_{i}(\v p) := \frac{V\left[\mathcal{T}^{\, \beta}_i (\v p)\right]}{V(\Delta_d)} \quad  \text{with} \quad i \in \{\emptyset, -, +\},
\end{equation}
where $V$ denotes the volume measured using the Euclidean metric. 

We start our analysis of volumes of thermal cones by presenting an operational interpretation in terms of guessing probabilities for the future and past of a given state subject to a thermal evolution; such interpretation provides a solid basis for presenting the aforementioned volumes as resource-theoretic monotones. Subsequently, we proceed to an in-depth analysis of the volumes with a particular focus on their behaviour as a function of the inverse temperature $\beta$. \rd{Finally, we explain how to modify the analysis to obtain meaningful volumes of entanglement cones.}


\subsection{Interpretation}

Consider a task of predicting the future, which roughly translates to guessing a state $\v{q}$ by having knowledge that it has originated from a given prior state $\v p$. In this case, the probability of correctly guessing a state that is $\epsilon$-distant from ${\v q}$ is given by
\begin{equation}
    \text{Pr}\left( {\v q}, \epsilon\,|\,\v q\in\mathcal{T}_+(\v p)\right) = \frac{V(\mathcal{B}_\epsilon(\v q))}{\mathcal{V}^{\, \beta}_+(\v{p})},
\end{equation}
where $\mathcal{B}_\epsilon(\v q)$ is an $\epsilon$-ball centred at $\v q$. We may get rid of the dependence on $\epsilon$ and $\v q$ by taking the ratio for two different states $\v p_1$ and $\v p_2$,
\begin{equation}
    \frac
    {\text{Pr}\left[\v q, \epsilon\,|\,\v q\in\mathcal{T}^{\, \beta}_+(\v p_2)\right]}
    {\text{Pr}\left[\v q, \epsilon\,|\,\v q\in\mathcal{T}^{\, \beta}_+(\v p_1)\right]} = \frac{\mathcal{V}^{\, \beta}_+(\v{p}_1)}{\mathcal{V}^{\, \beta}_+(\v{p}_2)} .
\end{equation}
Thus, the ratio of volumes yields a relative probability of guessing the future of two different states. In particular, if $\v p_2 \in \mathcal{T}^{\, \beta}_+(\v p_1)$, then also $\mathcal{T}^{\,\beta}_+(\v p_2)\subset\mathcal{T}^{\, \beta}_+(\v p_1)$ and in consequence
\begin{equation}
    \frac{\text{Pr}\left[\v q, \epsilon\,|\,\v q\in\mathcal{T}^{\, \beta}_+(\v p_2)\right]}
    {\text{Pr}\left[\v q, \epsilon\,|\,\v q\in\mathcal{T}^{\, \beta}_+(\v p_1)\right]} > 1. 
\end{equation}
This can be understood as follows -- as the evolution of a system progresses, the future becomes easier to guess or, in other words, more predictable.

In complete analogy, we may define a game in which, instead of guessing -- or predicting -- the future of a state $\v p$, one has to guess its past. For such a game, it is easy to show that 
\begin{equation}
    \frac
    {\text{Pr}\left[\v q, \epsilon\,|\,\v q\in\mathcal{T}^{\, \beta}_-(\v p_2)\right]}
    {\text{Pr}\left[\v q, \epsilon\,|\,\v q\in\mathcal{T}^{\, \beta}_-(\v p_1)\right]} = \frac{\mathcal{V}^{\, \beta}_-(\v{p}_1)}{\mathcal{V}^{\, \beta}_-(\v{p}_2)},
\end{equation}
so that the ratio of volumes yields the relative probability of guessing the past. Given $\v p_1\in \mathcal{T}^{\, \beta}_-(\v p_2)$ we have ${\mathcal{V}^{\, \beta}_-(\v{p}_1)}/{\mathcal{V}^{\, \beta}_-(\v{p}_2)}<1$, which simply means that as the evolution towards equilibrium progresses, one finds the past of a given state harder and harder to guess correctly.



\subsection{Properties}

We now show that the volumes of the thermal cones are thermodynamic monotones, i.e., functions of a state that decrease under thermal operations.
\begin{thm}\label{thm_monotne}
    The relative volumes $\mathcal V^{\, \beta}_{+}$ and $1-\mathcal V^{\, \beta}_{-} = \mathcal V^{\, \beta}_{+} + \mathcal V^{\, \beta}_{\emptyset}$ are thermodynamic monotones. Moreover, both monotones are faithful, taking the value 0 only when applied to the Gibbs state $\v \gamma$.
\end{thm}

\begin{proof}
One can straightforwardly show that both quantities decrease monotonically under thermodynamic operations. This is a simple consequence of the fact that for $\v p$ and $\v q$ connected via a thermal operation, we have \mbox{$\mathcal{T}^{\, \beta}_{+}(\v{q}) \subset \mathcal{T}^{\, \beta}_{+}(\v{p})$} and \mbox{$\mathcal{T}^{\, \beta}_{-}(\v{p}) \subset \mathcal{T}^{\, \beta}_{-}(\v{q})$} which automatically implies $\mathcal V^{\, \beta}_{+}(\v q)\leq \mathcal V^{\, \beta}_{+}(\v p)$ and $\mathcal V^{\, \beta}_{-}(\v q)\geq \mathcal V^{\, \beta}_{-}(\v p)$. 

In order to demonstrate its faithfulness, first note that every state that is not thermal can be mapped to a thermal state, thus showing that $1 - \mathcal{V}_-(\v \gamma) = 0$. Similarly, the Gibbs state cannot be mapped via thermal operations to anything else than itself, thus $\mathcal{V}_+(\v \gamma) = 0$. Now, in order to show that both monotones are non-zero for any state different from the Gibbs state, it suffices to demonstrate that $\mathcal{V}_+(\v p) > 0$ for any $\v p\neq \v \gamma$. It is enough to consider, without loss of generality, a state $\v p$ thermalised within a $(d-1)$-dimensional subspace, i.e., \mbox{${p_i}/{p_j} = {\gamma_i}/{\gamma_j}$} for all $i\neq j\neq1$. For $k\in\{2,\hdots,d\}$, using $\beta$-swaps between levels $1$ and $k$ defined as~\cite{Lostaglio2018elementarythermal}
$$
    \left\{p_1,\,p_k\right\} \longmapsto \left\{\frac{\gamma_1 - \gamma_k}{\gamma_1}p_1  + p_k,\,\frac{\gamma_k}{\gamma_1}p_1\right\}, 
$$
we generate $d-1$ new points shifted from the original state by displacements $(\v \delta_k)_i = \delta_k \epsilon_{ik}$ defined by the Levi-Civita symbol $\epsilon_{ik}$ and $\delta_k \neq 0$. The entire set $\{\v \delta_k\}_{k=2}^d$ of displacements is linearly independent; therefore, they define a $(d-1)$-dimensional simplex of non-zero volume. Conversely, if we assume that $\delta_k = 0$ for any $k$, we are led to a conclusion that the system is thermalised between levels $1$ and $k$ and, by transitivity, it must be equal to the Gibbs state $\gamma$, which leads to a contradiction with the initial assumption. Finally, we restore the full generality by noticing that any state $\v q$ contains in its future cone states thermalised in any of the $(d-1)$-dimensional subspaces and, consequently, their entire future cone. Therefore, the non-zero volume of the latter implies the non-zero volume of the former. 
\end{proof}

\begin{figure}[t]
    \centering
    \includegraphics{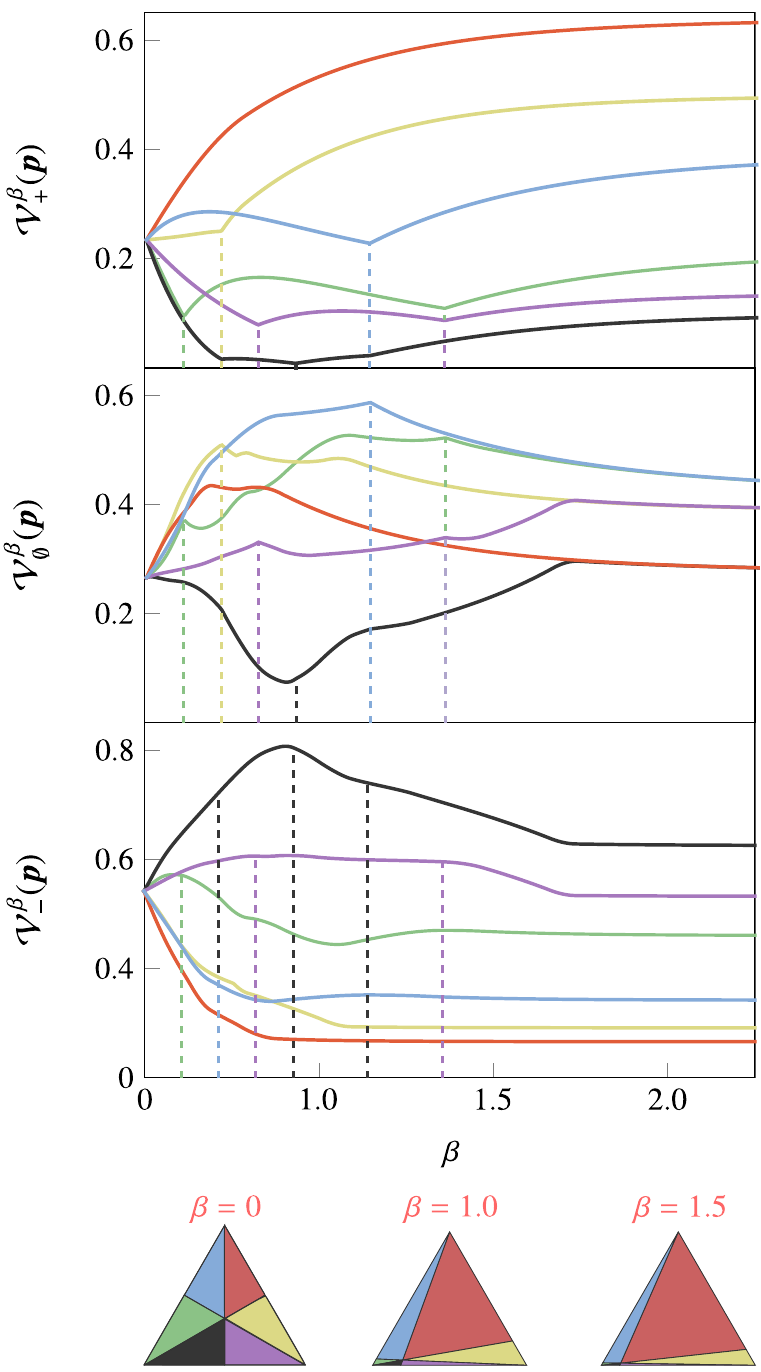}
    \caption{\label{fig-volumeasfunctionofbeta} \emph{Volume of the thermal cones as a function of $\beta$}. For all permutations of the state $\v p = (0.52, 0.12, 0.36)$, we plot the volume of the future thermal cone $\mathcal{V}^{\, \beta}_{+}$ (top), the incomparable thermal region $\mathcal{V}^{\, \beta}_{\emptyset}$ (centre) and the past thermal cone $\mathcal{V}^{\, \beta}_{-}$ (bottom). Each colour corresponds to a permutation associated with a given chamber of the probability simplex. Among all states, two are distinct: the maximally active (red curve) $\v p_{\text{max}} = (0.12, 0.36, 0.52)$ and the passive (black curve) $\v p_p = (0.52, 0.36, 0.12)$ states. Any other permutation of the initial state characterises a different active state. \rd{The kinks in $\mathcal{V}^{\, \beta}_{+}$ match with the inverse temperatures at which $\v p$ changes its $\beta$-ordering (vertical lines of matching colors).} The three different simplices at the bottom show how the Weyl chambers change with~$\beta$.}
\end{figure}

The behaviour of $\mathcal V_+$ and $\mathcal V_-$ as a function of $\beta$ strongly depends on the $\beta$-ordering of the state under consideration. Among all $\beta$-orderings, there are two extreme cases, namely the one where the population and energies are arranged in non-increasing order ($p_i \leq p_j$ for $E_i \leq E_j$), and the other in which the populations are `anti-ordered' with respect to the energies ($E_i \leq E_j$ implies $p_i \geq p_j$). These two distinct $\beta$-orderings characterise \emph{passive} and \emph{maximally active} states, respectively~\cite{Pusz1978,Lenard1978}. At the bottom left of Fig.~\ref{fig-volumeasfunctionofbeta}, for a three-level system, we depict each Weyl chamber with different colours: passive states lie in the black chamber, whereas maximally active states lie in the red one.

Let us first focus on the volume of the future thermal cone and analyse how it changes by increasing the inverse temperature $\beta$ from $\beta = 0$ to $\beta \to \infty$. In the present analysis, the initial state is kept fixed, while the thermal state is taken to be a function of temperature $\v \gamma = \v \gamma(\beta)$, and it follows a trajectory from the centre of the probability simplex to the ground state (see Fig.~\ref{fig-thermal-cones-examples} for an example considering a three-level system). \rd{If at $\beta = 0$ the initial state $\v p$ is passive}, the volume of its future thermal cone first decreases with $\beta$, and then starts to increase when $\v \gamma$ passes $\v p$ (\rd{i.e., when $\v{p}$ changes its $\beta$-ordering}), tending asymptotically to a constant value (see black curves in Fig.~\ref{fig-volumeasfunctionofbeta}). However, \rd{if at $\beta = 0$ the initial state $\v{p}$ is maximally active}, the behaviour of the volume of $\T_+^\beta(\v{p})$ differs from the previous case. As the thermal state approaches the ground state with increasing $\beta$, the distance between maximally active states and $\v \gamma$ increases with $\beta$, because the ground state and $\v{p}$ are located in opposite chambers. Consequently, the volume increases asymptotically to a constant value (see red curves in Fig.~\ref{fig-volumeasfunctionofbeta}). \rd{For general states, one can provide a qualitative explanation of the behaviour of their volumes based on their $\beta$-orderings. The inverse temperatures $\beta$ for which one finds kinks in the future volume $\mathcal{V}_+^\beta(\v{p})$ (vertical lines in Fig.~\ref{fig-volumeasfunctionofbeta}) match with the transitions from a given $\beta$-ordering to another one. This should be compared with the isovolumetric level sets in Fig.~\ref{fig-equivolume}, where a matching non-smooth behaviour is found. Observe that similar kinks are not encountered for the past volume $\mathcal{V}^\beta_-(\v{p})$, related to the smooth behaviour of the corresponding level sets. However, by} considering a passive state $\v p$ and all its permutations, we can demonstrate that maximally active and passive states have maximum and minimum future volumes, respectively: 
\begin{cor}
\label{obs_passive_active_obs}
For a $d$-dimensional energy-incoherent state $\v p$ with Hamiltonian $H$, and all states defined by permuting its population, the future thermal cone of the permutation resulting in the maximally active and passive states achieve maximum and minimum volumes, respectively.
\end{cor} 
\begin{proof}
The Corollary is proven by noting that all permutations of $\v p$ are thermomajorised by the one corresponding to the maximally active state $\v p_{\text{max}}$, while the associated passive state $\v p_{p}$ is thermomajorised by all the other permutations. Consequently, $ \mathcal{T}_{+}(\v p_p) \subset \mathcal{T}_{+}(\Pi \v p)$, while $\mathcal{T}_{+}(\Pi \v p) \subset \mathcal{T}_{+}(
\v p_{\text{max}})$ for any permutation matrix $\Pi$.
\end{proof}
Corollary~\ref{obs_passive_active_obs}, also implies that the volume of the past thermal cone is minimum and maximum for the one corresponding to the maximally active and passive states, respectively (see~Fig.~\ref{fig-volumeasfunctionofbeta}). To provide further characterisation of the volumes, we apply Lemma~\ref{lemma_incomparablefiniteT} to a non-full rank state, which allows us to derive the following result: 
\begin{cor}
\label{cor_thermalvolumepast}
The past thermal cone of a non-full rank state has volume zero despite being non-empty. 
\end{cor} 
\begin{proof}
Without loss of generality, consider a state of non-full rank $\v{p} = (p_1, ..., p_{d-1},0)$. Applying Eq.~\eqref{eq_thermaltangentvectors} yields $\v{t}^{(d,\v \pi)} = (1,0,...,0)$, for all $\v \pi$. Consequently, the incomparable region is given by all points in the interior of the probability simplex, except those that are in the future of $\v{p}$. Then, according to Theorem \eqref{eq_incomparableconeFT}, all the points of the past will be located at the edge, and therefore the volume of $\T^{\, \beta}_{-}(\v{p})$ is zero.
\end{proof}

Understanding the behaviour of the thermal incomparable region is not directly straightforward. However, Corollary~\ref{cor_thermalvolumepast} helps us to find the state of non-full rank with the largest incomparable region:
\begin{cor}
\label{cor_full_rank_largest}
The non-full rank state with the largest thermal incomparable region is given by
\begin{equation}
\label{eq_nonfullrankstatethermalinc}
\v g = \frac{1}{Z}\left(e^{-\beta E_1}, ...,e^{-\beta E_{d-1}},0 \right) \quad \text{where} \quad Z = \sum_{i=1}^{d-1}e^{-\beta E_i}. 
\end{equation}
\end{cor}
\begin{proof}
Consider an arbitrary non-full rank state $\v{p}$. According to Corollary~\ref{cor_thermalvolumepast}, the volume of the thermal incomparable region can be written as $\mathcal V^{\, \beta}_{\emptyset}(\v p) = 1-\mathcal V^{\, \beta}_+(\v p)$. Now note that, $\v p \succ_{\beta} \v g$, and $\mathcal T_{+}(\v g) \subset \mathcal T_{+}(\v p)$. This implies that $\v g$ is the non-full rank state with the smallest future
thermal cone and, therefore, with the largest incomparable region. 
\end{proof}

\begin{figure}[t]
    \centering
    \includegraphics{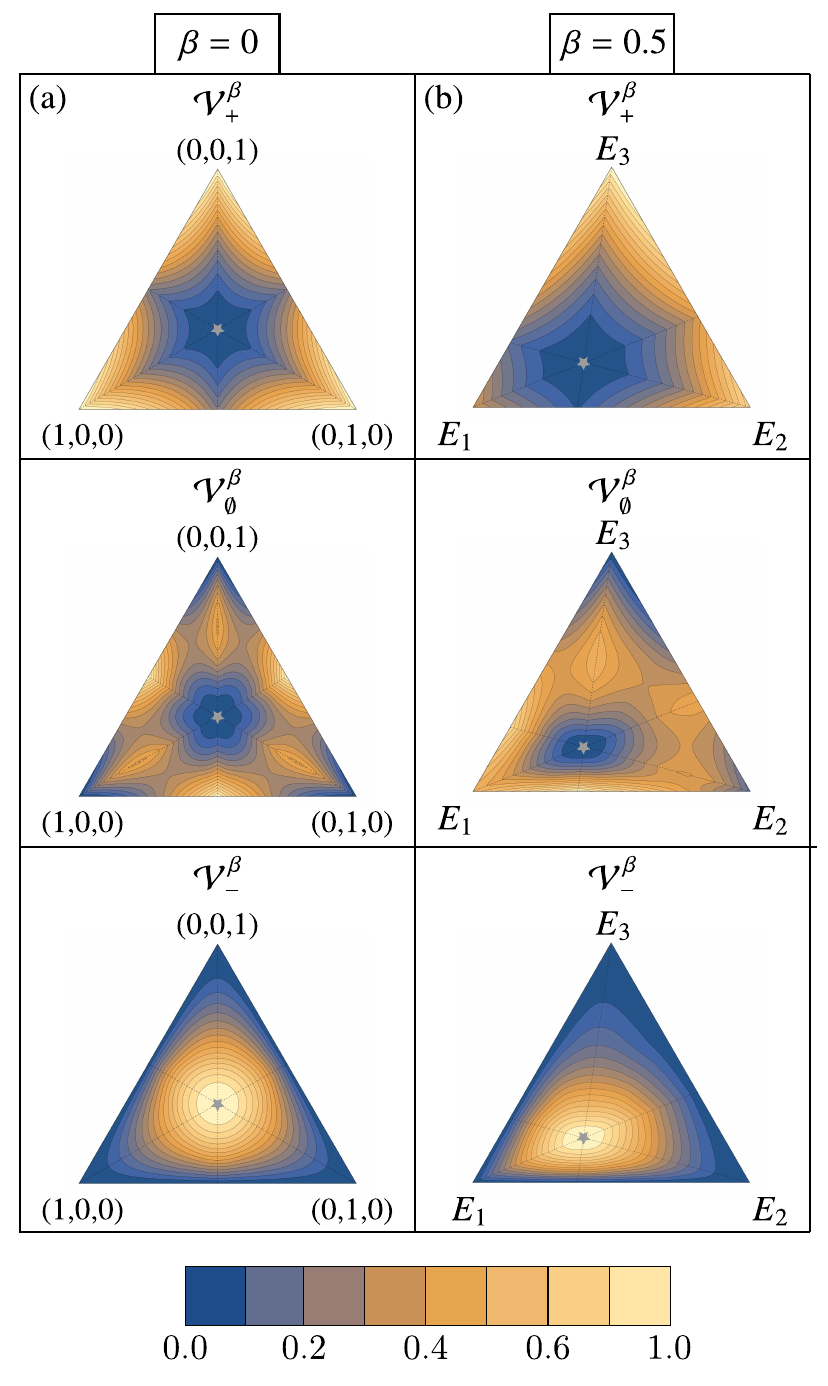}
    \caption{\label{fig-equivolume} \rd{\emph{Isovolumetric curves}. The volume of the future $\mathcal V^{\, \beta}_{+}$, incomparable $\mathcal V^{\, \beta}_{\emptyset}$, and past thermal regions $\mathcal V^{\, \beta}_{-}$, in the space of three-dimensional probability distributions for an equidistant energy spectrum $E_1 =0$, $E_2=1$ and $E_3 = 2$ and inverse temperature (a) $\beta =0$ and (b) $\beta = 0.5$. The thermal state is depicted by a grey star {$\color{gray}\bigstar$}.} 
    }
\end{figure}

So far, the behaviour and properties of the volumes have been analysed and discussed without explicitly calculating them. There are several known algorithms for computing volumes of convex polytopes, such as triangulation, signed decomposition methods, or even direct integration~\cite{iwata1962, buller1998}. These algorithms can be employed to obtain the volumes of past and future cones; the volume of the incomparable region can be calculated using the fact that the total volume of the probability simplex $\Delta_d$ is equal to one. For a three-dimensional energy-incoherent state $\v p$, expressions for the volumes can be easily derived. The starting point is to consider the Gauss area formula~\cite{braden1986}, which allows us to determine the area of any polygon with vertices described by Cartesian coordinates. Taking into account a polygon $P$ whose vertices, assumed to be arranged along the boundary in a clockwise manner, are denoted by $P_i = (x_i, y_i)$, with $i= \{1, ..., n\}$, the Euclidean volume can be expressed as 
\begin{equation}
V = \frac{1}{2} \left| \sum_{i=1}^n \operatorname{det} \begin{pmatrix} x_i & x_{i+1} \\ y_i & y_{i+1} \end{pmatrix} \right| ,   \end{equation}
where $x_{n+1} = x_1$ and $y_{n+i} = y_1$. For $\beta = 0$, deriving the volume of the thermal cones is straightforward, since the vertices, or extreme points, are permutation of $\v p$. In this case, we arrive at the following closed-form expressions:
\begin{align}
    \mathcal V^{\, 0}_{+}(\v p) =& \: (3p^{\downarrow}_1-1)^2-3(p^{\downarrow}_2-p^{\downarrow}_1)^2 , \label{eq_futurevolume} \\ 
    \mathcal V^{\, 0}_{\emptyset}(\v p) =& \: 1-3(1-p^\downarrow_1)^2+ (1-3p^\downarrow_3)^2 - 2 \mathcal V^{\, 0}_{+}(\v p) \\
    &\:+ 3\theta(1/2 - p^\downarrow_1)(1-2p^\downarrow_1)^2 \nonumber \\ 
    \mathcal{V}^{\, 0}_{-}(\v p) =& \: 12p^\downarrow_2p^\downarrow_3 - 3\theta(1/2 - p^\downarrow_1)(1-2p^\downarrow_1)^2 , \label{eq_pastvolume} \nonumber
\end{align}
where $\theta$ is the Heaviside step function. The situation involving a finite $\beta$ is not as simple as before. Now, the extreme points are obtained by applying Lemma~\ref{lem_extreme}, and, although computationally it is an easy task to calculate them, a neat and concise closed-form expression cannot be derived. 
\begin{figure*}[t]
    \centering
    \includegraphics{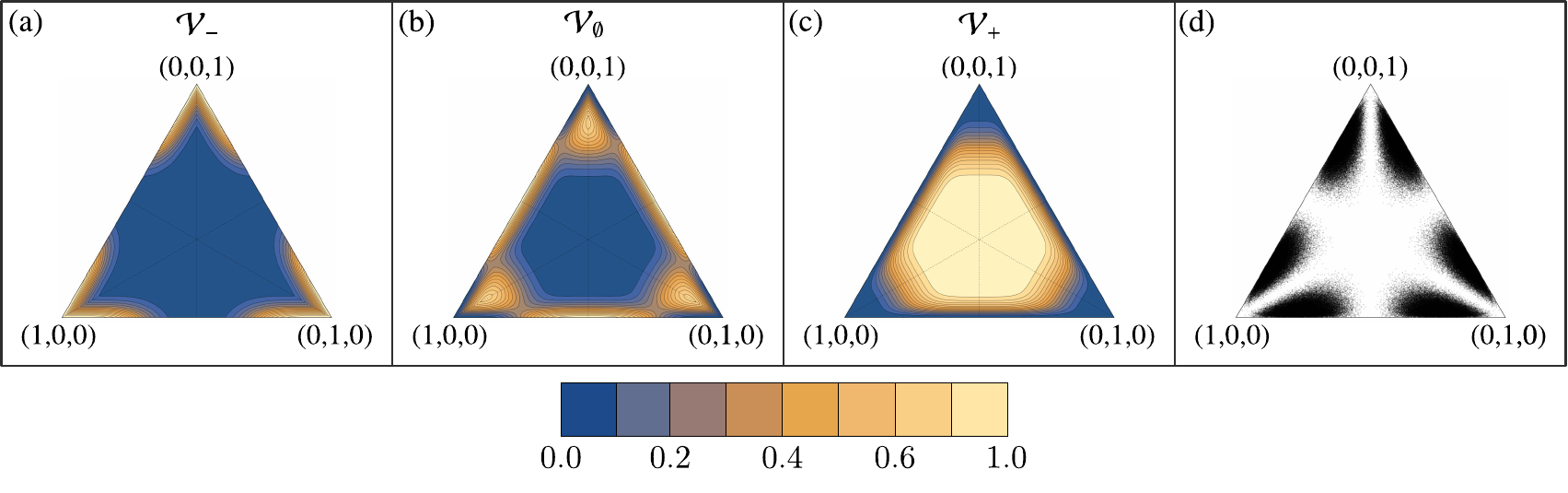}
    \caption{\label{fig-level-sets-entanglement}\emph{Isovolumetric sets for entanglement cones}. \rd{In order to calculate the isovolumetric lines for (a) the past cone, (b) the incomparable region and (c) the future cone for entanglement resource theory of states in $\mathcal{H}^{3\times3}$ one considers the density of the Schmidt coefficients $\v{p}$ induced by the Haar measure in the space $\mathcal{H}_9$ of pure bipartite states, as depicted in panel (d) by showing a sample of $5\cdot10^4$ points in $\Delta_3$. Observe the scarcity of points in the central region, resulting in characteristic concentration of states with large future in the centre of the simplex and large past in the vicinity of the vertices.}} 
\end{figure*}

Finally, let us look at iso-volumetric curves for different values of $\beta$. In Fig.~\ref{fig-equivolume}, we plot these curves for a three-level system and four different temperatures. As expected, the symmetry is broken for any $\beta > 0$ as $E_3 > E_2 > E_1$. A simple fact worth mentioning is that the volume of the future thermal cone of the highest excited state $\v s_d = (0,\hdots,0,1)$ is always independent of $\beta$, maximum, and equal to unity. Conversely, the past thermal cone volume is maximum for a Gibbs state and equal to unity. Moreover, these curves give insight into how resourceful states are distributed within the space of states. 

\subsection{\rd{The volumes of entanglement cones}}

\rd{Finally, we will briefly discuss the general qualitative aspects of the volumes of entanglement cones based on the numerical considerations. Detailed formal methods used in order to obtain them are discussed in Appendix \ref{app:entanglement_vols}.}

\rd{Naturally, depending on the context, the probability distribution $\v{p}$ may pertain to the Schmidt coefficients of a pure entangled state or the coefficients resulting from decomposing the state in a distinguished basis in the context of coherence~\cite{Streltsov2016,Streltsov2017}. 
Despite the close connection between the resource theory of thermodynamics (at $\beta = 0$) with the resource theories of entanglement and coherence, crucial differences appear already at the level of a single state as the order is reversed~\cite{zyczkowski2002}, hence, interchanging the future $\mathcal{T}_+$ with the past $\mathcal{T}_-$. The difference is even more pronounced within the context of volumes for the entanglement of pure states under LOCC operations.

The distribution of Schmidt coefficients induced by the uniform Haar measure in the space of pure bipartite states is significantly different from the flat distribution in the probability simplex $\Delta_d$ \cite{zyczkowski2001induced}. In particular, one observes a repulsion from the centre and, in the case of entangled systems of unequal dimension, from the facets of the simplex. Consequently, this implies a significant difference in the qualitative features of the isovolumetric curves. Fig.~\ref{fig-level-sets-entanglement} shows the isovolumetric curves for $d=3$ with equal-sized systems. Observe that states with large future volumes $\mathcal{V}_+$ are concentrated around the centre of the simplex, which is explained by the repelling property. Inverse effects can be seen for the states with large past volumes $\mathcal{V}_-$, which concentrate at the boundaries of the simplex. The differences become even more pronounced for systems of unequal dimensions, as we demonstrate in qualitative figures in Appendix~\ref{app:entanglement_vols}~(see Fig.~\ref{fig-entanglementvolume-2}).}


\section{Outlook} 
\label{Sec:Outlook}

In this paper, we have investigated the structure of the thermodynamic arrow of time by dividing the space of energy-incoherent states into the past, the future, and the incomparable region, in analogy with the future, past, and spacelike regions of Minkowski spacetime, respectively. These regions of the probability simplex, called thermal cones, encode in a natural way the achievability of state transformations under thermal operations. In particular, we considered the energy-incoherent states in the presence of a thermal bath at a finite temperature and in the limit of temperature going to infinity. The incomparable and past thermal cones were fully characterised and carefully analysed in both regimes. Furthermore, we identified the volumes of the thermal cones as thermodynamic monotones and performed detailed analysis of their behaviour under different conditions. Our results can also be applied directly to the study of entanglement, as the order defined on the set of bipartite pure entangled states by local operations and classical communication is the opposite of the thermodynamic order in the limit of infinite temperature. In this context, the future thermal cone becomes the past for entanglement, and the past becomes the future. Moreover, a similar extension to coherence resource theory can also be drawn.

There are quite a few research directions that one may want to take to generalise and extend the results presented in this paper. \rd{First, one may generalise our analysis beyond the energy-incoherent states to the full space of quantum states, where the available tools are comparatively scarce \cite{Korzekwa2017, LostaglioKorzekwaCoherencePRX, Gour_2018}}.
\rd{As explained in Appendix~\ref{app:coherent_thermal}, for qubit systems one can employ different known techniques to construct coherent thermal cones under both Gibbs-preserving and thermal operations, as shown in Fig.\ref{fig-coherent-thermal-cone}.} 
\begin{figure}[t]
    \centering
    \includegraphics{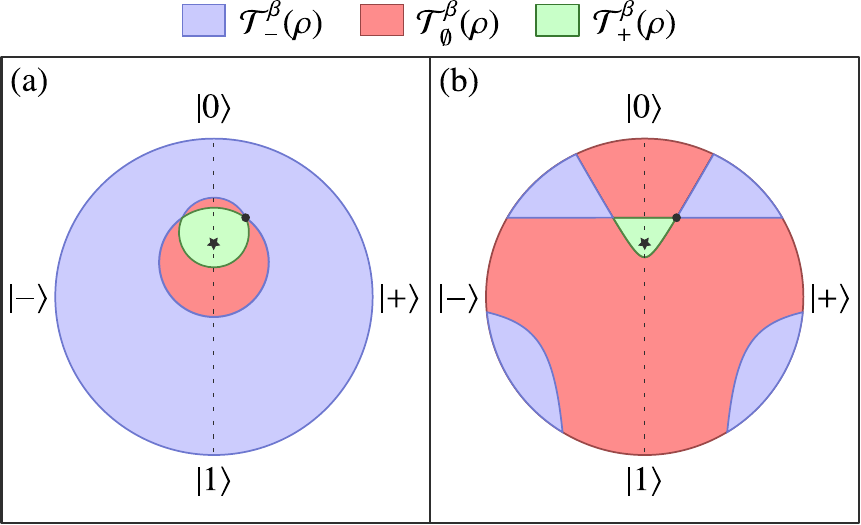}
    \caption{\label{fig-coherent-thermal-cone} \rd{\emph{Coherent thermal cones.} For a two-level system with initial state $\rho$, represented by a black dot $\bullet$, and thermal state represented by a black star $\bigstar$  with Bloch vectors \mbox{$\v r_{\rho} = (0.2,0,0.5)$} and \mbox{$\v r_{\gamma} = (0,0,1/3)$}, respectively, we depict in the real cross-section of the Bloch ball, the coherent thermal cone (a) under Gibbs-preserving operations (b) under thermal operations}.
    }
\end{figure}
\rd{Beyond the qubit case, the problem is still open. \footnote{\rd{Note that this problem was formally solved in Ref.~\cite{Gour_2018}, in which a complete set of necessary and sufficient conditions was found for arbitrary quantum state transformations under thermal operations. However, this set of conditions is convoluted and hard to characterise, thus beyond implicit formulation, the problem remains open.}}}

Second, a newly introduced notion of continuous thermomajorisation was presented in Ref.~\cite{lostaglio2021continuous,Korzekwa2022}, where the authors obtained the necessary and sufficient conditions for the existence of a Markovian thermal process underlining a given state transformation. A natural task then is to provide an equivalent construction of the past and incomparable regions for that memoryless thermodynamic ordering. Finally, we point out a possible technical extension of our results. Our investigation was performed in the spirit of single subsystems, but a possible extension comprising many non-interacting subsystems (possibly independent and identically distributed) could be done by defining a proper function to analyse the behaviour of the thermal cones.

\begin{acknowledgments}
    AOJ thanks Moisés Bermejo Morán for helpful discussions. We acknowledge financial support from the Foundation for Polish Science through the TEAM-NET project (POIR.04.04.00-00-17C1/18-00), and the NCN projects:
    DEC-2019/35/O/ST2/01049 and 2021/03/Y/ST2/00193.
\end{acknowledgments}


\appendix


\section{Derivation of the incomparable region -- Proof of Lemmas \ref{lem_incomparablecone} and \ref{lemma_incomparablefiniteT}} 
\label{sec_appendix}

In this Appendix, we consider lemmas defining the boundaries of the incomparable region -- an essential step in learning the full form of the past thermal cone. First, we consider the case of $\beta = 0$, covered by Lemma \ref{lem_incomparablecone}, and develop the notion of the so-called tangent vectors $\v t^{(n)}$ comprising the boundary of the incomparable region. The notion of tangent vectors is generalised and used in the proof of Lemma \ref{lemma_incomparablefiniteT} concerning the incomparable region for $\beta > 0$.

\subsection{Infinite temperature}

\subsubsection{Tangent vector}

The first notion we will need in the proof is that of a \textbf{tangent probability vector}, called the tangent vector for short, which will prove to be an essential ingredient.

Let $\Delta_d$ be the set of probability vectors of dimension $d$ with real entries, $\Delta_d= \left\{(p_1, ..., p_d) \in \mathbb{R}^d : \sum_i p_i = 1\right\}$,  and let us restrict ourselves to vectors ordered in a non-increasing order, i.e., \mbox{$p_i \geq p_{i+1}$}. In order to avoid complication we assume, without loss generality, that all $p_i \neq p_j$. In what follows, we will denote probability vectors by bold lowercase letters $\v p \in \Delta_d$ and their corresponding cumulative counterparts by bold uppercase vectors $\v P: P_i = \sum_{j=1}^i p_i$ with $i\in\{0,\hdots,d\}$. For any vector $\v p$ we introduce a tangent vector $\v{t}(\v{p}) \equiv \v{t}\in\Delta_d$ by imposing that all its components except the first and the last are equal, $t_i = t_j$ for all $1 < i < j < d$; furthermore, we require that the cumulative vector $\v{T}$ agrees with the vector $\v{P}$ in at most two consecutive points, i.e., $T_i = P_i$ and $T_j > P_j$ for all $j\in\left\{1,\hdots,d-1\right\}\backslash\{i\}$ or $T_j = P_j$ for $j \in\left\{i, i+1\right\}$ and $T_j > P_j$ elsewhere. The two imposed conditions follow the intuition of tangency and, by construction, satisfy the majorisation relation $\v t \prec \v p$. Furthermore, note that $T_0 = P_0 = 0$ and $T_d = P_d = 1$ and thus are naturally excluded from the considerations.

Indeed, assuming equality between $\v T$ and $\v P$ at exactly two consecutive points restricts the tangent vectors to a set of $d$ unique probability vectors $\v{t}^{(n)}$ defined as 
\begin{equation}
\label{A:tvector}
    \v{t}^{(n)}(\v{p})\equiv \v{t}^{(n)}= \left(t^{(n)}_1, p_n, ..., p_n, t^{(n)}_d\right) , 
\end{equation}
for $1\leq n \leq d$ with the first and last components given by 
\begin{equation}
    t_1^{(n)}  = \sum_i^{n-1} p_i - (n-2) p_n\quad, \quad
    t_d^{(n)}  = 1 - t^{(n)}_1 - (d-2)p_n.
\end{equation}

Observe that the tangent vectors $\v{t}^{(n)}$ that agree with the Lorenz curve of $\v p$ at two successive points can be used to construct all possible tangent vectors $\v t$ that satisfy the condition of agreement at at least a single point, $T_i = P_i$. Indeed, consider a vector $\v t (\lambda, i) \equiv \v t= (1 - \lambda)\v{t}^{(i)} + \lambda \v{t}^{(i+1)}$. Direct calculation shows that it is tangent at just one point, $T_i = \lambda T_i^{(i)} + (1-\lambda)T_i^{(i+1)} = P_i$. Similarly we arrive at $T_j > P_j$ for $j \neq i$ and $0<\lambda<1$. 

Thus, starting with a discrete set of $d$ tangent vectors $\v{t}^{(n)}$ we recover the entire continuous family of tangent vectors $\v{t}(\lambda, i)$. Indeed, this argument can be further formalised by considering the Lorenz curves $f_{\v p}(x)$ and $f_{\v t}(x)$. Considering the left and right derivatives of the former, $\lim_{x \rightarrow \frac{i}{d}_-} f'_{\v p}(x) = d\cdot p_i$ and $\lim_{x \rightarrow \frac{i}{d}_+} f'_{\v p}(x) = d\cdot p_{i+1}$, we obtain the extremal slope values for the tangent lines at the $i$-th elbow. Now, considering the second Lorenz curve by construction we have that $f_{\v t}(i/d) = f_{\v p}(i/d)$ and its derivative at this point, $f'_{\v t}(i/d) = d\left[(1-\lambda)p_i + \lambda p_{i+1}\right] $, span all values between the extremal slope values $d\cdot p_i$ and $d\cdot p_{i+1}$, therefore exhausting the family of possible tangent lines at the $i$-th elbow. 

\begin{figure*}
    \centering
    \includegraphics{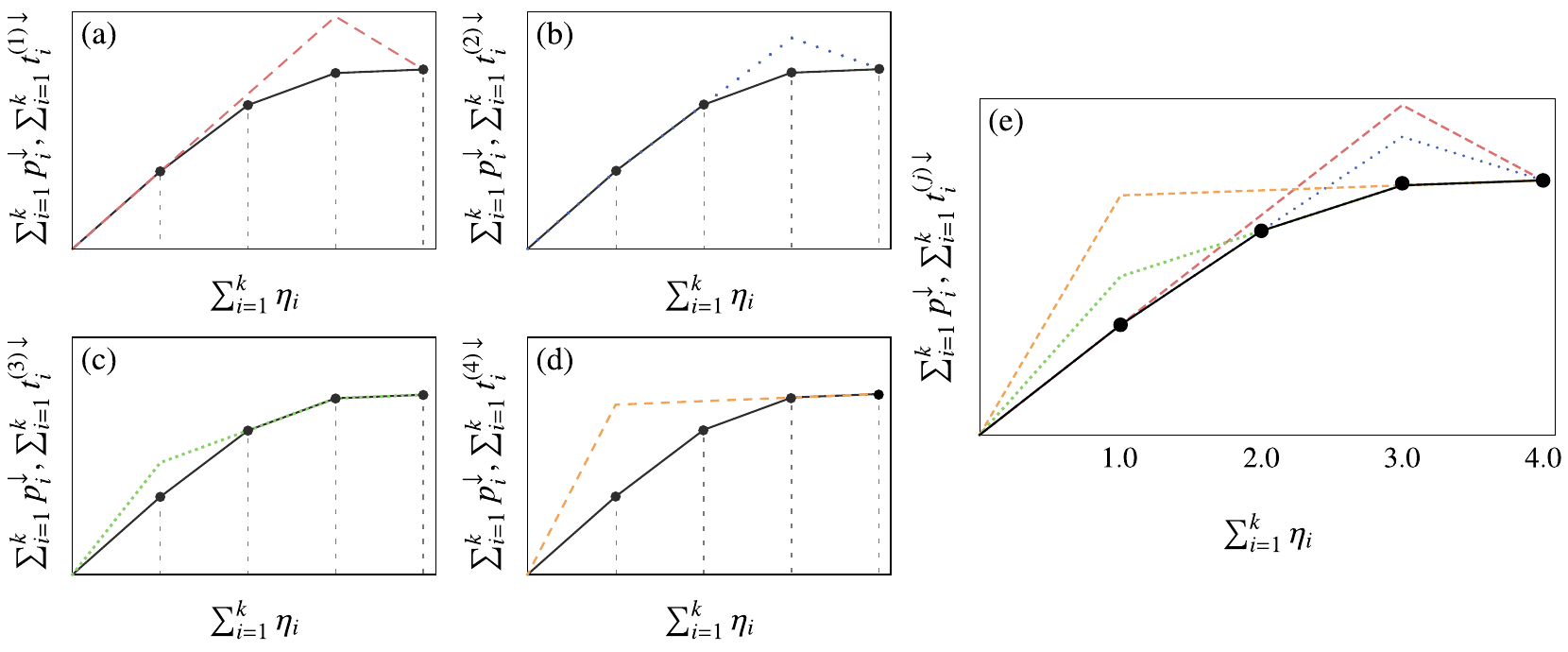}
    \caption{\label{fig-majorisationd4-examples2} \emph{Lorenz curves of the tangent vectors $\v t^{(n)}$}. Majorisation curves of the state $\v p = (0.43, 0.37, 0.18, 0.02)$ (black) and (a) $\v{t}^{(1)}$ (b)  $\v{t}^{(2)}$ (c) $\v{t}^{(3)}$ (d) $\v{t}^{(4)}$ and (e) all tangent vectors $j \in \{1,2,3,4\}$, respectively.} 
\end{figure*}

\subsubsection{Lattice}

Lattices provide a setting within which it is natural to represent the precedence or succession of elements within a given set. In particular, they can be used to equip a given set with a time-like structure, as provided by the definition of a lattice (see Definition~\ref{def_Lattice}).

It is well known that the partially ordered set $(\Delta_d, \succ)$ of $d$-dimensional probability vectors with real entries in non-increasing order under majorisation forms a lattice~\cite{cicalese2002,Korzekwa2017}. In particular, within this setting, the join $\v p \vee \v q$ can be interpreted as the last common past point of $\v p$ and $\v q$. Similarly, the meet $\v p \wedge \v q$ can be seen as the first common future point of the pair $\v p, \v q$. The procedure to obtain the join and meet has been illustrated in Ref.~\cite{Korzekwa2017}, and since part of our proof relies on the existence of the join, we will now review the algorithm used to construct it.  

To construct the join of $\v p$ and $\v q$, we start with a probability vector $\v r^{(0)}$ with elements defined by
\begin{equation}
    r_i^{(0)} = \max\left\{\v P_i, \v Q_i\right\} - \max\left\{\v P_{i-1}, \v Q_{i-1}\right\} .
\end{equation}
At this stage, it may occur that the entries of $\v r^{(0)}$ are not ordered non-increasingly. However, it is possible to arrive at a properly ordered probability vector $\v r = \v p \vee \v q$ defining the actual join in no more than $d - 1$ steps. In each step $k\geq 0$ we define $N \geq 2$ as the smallest point of increase between two consecutive components of the probability vector $r^{(k)}$, that is, $r^{(k)}_N > r^{(k)}_{N-1}$. Next, $M \leq N - 1$ is defined in such a way that by introducing constant probabilities for the entries with $i\in\{M,\hdots,N\}$ the growth is eliminated. It is done by requiring that
\begin{equation} \label{eq:joinConstruct_aCoeffs}
    r^{(k)}_M \geq \frac{\sum_{i = M}^N r_i^{(k)}}{N - M + 1} =: a_k.
\end{equation}

Thus, the next iterative step $\v r^{(k+1)}$ is defined by setting its components as

\begin{equation}\label{eq:joinConstruct_flatFragments}
    r_i^{(k+1)} = 
    \begin{cases}
        a_k & \text{for } i\in\{M,\hdots,N\} \\
        r_i^{(k)} & \text{otherwise}
    \end{cases}.
\end{equation}
This construction is repeated until for some $k'$ the probability vector $\v r^{(k')} \equiv \v r$ is ordered non-increasingly; in this way we get the proper join.

\subsubsection{Incomparable region and the boundaries}

As a final piece of information needed to understand the proofs, we introduce the definition of the boundary of the past cone.

\begin{defn}[Boundary of the past thermal cone] \label{def_pastBoundary}
Consider a $d$-dimensional energy incoherent state $\v p\in\mathcal{P}_d$ with $d \geq 3$. We define the boundary of the past thermal cone as the set of probability vectors $\v q \succ \v p$ for which the cumulative vector $\v Q$ is equal to the cumulative vector $\v P$ at least one point. In other words, $\v P_j < \v Q_j$ for some proper subset of the indices $j$ and $\v P_i = \v Q_i$ for all other indices.
\end{defn}
One may define the boundary of the future cone in a similar way by changing the direction of the inequalities. By considering the common part of the boundary between future and past, one comes to a simple observation:

\begin{obs}[Common point of future and past cones]
    A point that lies simultaneously at the common part of the boundary between the future and the past must fulfil $\forall_i P_i = Q_i$. Therefore, for $\beta = 0$ we have the equality of $\v p$ and $\v q$ up to a permutation, giving a total of $d!$ common points between the future and the past of any vector $\v p$.
\end{obs}

Equipped with the notion of tangent vectors $\v t^{(n)}$, the join $\v p \vee \v q$ and the boundary of the past cone, we are now prepared to tackle the Lemma~\ref{lem_incomparablecone} concerning the incomparable region, and this is done by proving the following result:
    \begin{lem}
    \label{lem-incomp}
      Consider $\v{p}, \v{q} \in \Delta_d$ and assume that $\v p \nsucc \v q$. Then, $\v{q}$ belongs to the incomparable region of $\v{p}$, $\v q \in \mathcal{T}_\emptyset(\v p)$, if and only if it belongs to the future majorisation cone of some vector $\v t$ tangent to $\v p$, $\v{q}\in\T_+(\v{t})$, with
        $$
            \v{t} \equiv \v{t}(\v{p};\lambda, n)=  \lambda \v{t}^{(n)}(\v{p}) + (1-\lambda) \v{t}^{(n+1)}(\v{p}) ,
        $$
        for some $n\in \{1,d-1\}$ and $\lambda \in[0,1]$.
    \end{lem}
    \begin{proof}
    
To prove the ``if" direction, we take a probability vector lying in the interior of the future of the tangent vector, $\v{q} \in \operatorname{int} [\T_{+}(\v{t})]$. Then, to prove that $\v q \in \T_{\emptyset}(\v p)$, one needs to show that $\v{q} \nsucc \v{p}$. By construction, we have $T_k \geq  P_k$, for every $k\neq n$, with equality when $k = n$, so $P_n = T_n > Q_n$ with the equality excluded by putting $\v q$ in the interior of the future.  Consequently, $P_n > Q_n$, and thus $\v{q} \nsucc \v{p}$ and by the initial assumption, $\v p \nsucc \v q$. Therefore, $\v{q} \in \T_{\emptyset}(\v{p})$.

In order to demonstrate the ``only if" direction, let us take an arbitrary $\v q \in \T_{\emptyset}(\v p)$ and recall that there always exists the last common past point for $\v p$ and $\v q$ called join, $\v r = \v p \vee \v q$. From the construction of the join $\v r$, it is found that the entries of the cumulative distribution $\v R$ will be divided into three subsets, namely points common with $\v P$, common with $\v Q$ and the ones lying above either, that is $I_p := \left\{i:0<i<d,\,R_i = P_i\right\}$, $I_q := \left\{i:0<i<d,\,R_i = Q_i\right\}$ and $J := \left\{0<j<d,\,:R_j > \max(P_j,Q_j)\right\}$, respectively. In particular, looking at equation \eqref{eq:joinConstruct_flatFragments} one can see that it is either the case that $M-1\in I_p$ and $N+1\in I_q$, or $M-1\in I_p$ and $N+1\in I_q$ or by invoking geometric intuition, endpoints of the flat fragments of $\v R$ will join $\v P$ and $\v Q$. Thus, at each step, the sets $I_p$ and $I_q$ will be non-empty. 
By this argument, we may choose any index $i\in I_p$ and construct a tangent vector $\v t'$ for the join $\v{r}$ at the $i$-th elbow.
\begin{equation}
    \v t' \equiv \v t(\v{r};\mu,i) = \mu \v t^{(i)}(\v r) + (1 - \mu) \v t^{(i+1)}(\v r)
\end{equation}
for any $\mu\in[0,1]$. This vector obeys, by construction, the majorisation relation $\v t' \succ \v r \succ \v q$. Furthermore, due to the choice $i\in I_p$ it is also a tangent vector for the $\v p$,
\begin{equation}
    \v t' = \v t \equiv \v t(\v{p};\lambda(\mu),i) =  \lambda(\mu) \v t^{(i)}(\v p) + (1 - \lambda(\mu)) \v t^{(i+1)}(\v p)
\end{equation}
for $\lambda(\mu) \in [0,1]$. Therefore, by the properties of the tangent vector $\v t$ it follows that $\v q \in \T_{+}(\v t)$.  
\end{proof}

The last step necessary in order to demonstrate the Lemma~\ref{lem_incomparablecone} is to notice that the tangent vectors $\v{t}(\v{p};\lambda(\mu),i)$ are convex combinations of $\v t^{(i)}(\v p)$ and $\v t^{(i+1)}(\v p)$  and their future majorisation cones are convex, therefore, union of their future cones corresponds to the convex hull of the future cones of the extreme points
\begin{equation}
     \bigcup_{\lambda\in[0,1]}\mathcal{T}_+(\v{t}(\v{p};\lambda,i)) = \text{conv}\left[\mathcal{T}_+\left(\v t^{(i)}\right)\cup\mathcal{T}_+\left(\v t^{(i+1)}\right)\right]
\end{equation} 
which completes the statement of Lemma~\ref{lem_incomparablecone}.\qed

\subsection{Finite temperatures}
\label{app:finite_temp_derivs}

The proof of Lemma~\ref{lemma_incomparablefiniteT} for $\beta > 0$ is developed in the simplest way by considering the embedding $\mathfrak{M}$ introduced in Section~\ref{sec_embeddingLattice}, which takes $d$-dimensional probability distributions $\v p\in\Delta_d$ to its higher-dimensional image $\mathfrak{M}(\v p) \in \Delta^\mathfrak{M}_{d}$
which allows us to closely follow the steps of the proof for $\beta = 0$.

After shifting our focus to the embedded space, we construct the corresponding tangent vectors $\v{t}^{\mathfrak{M}}(\v{p};\lambda,n) = \lambda \v t^{\mathfrak{M}(n)}(\v p) + (1 - \lambda) \v t^{\mathfrak{M}(n+1)}(\v p)$ that respect the rules for constructing the majorisation curve within the embedded space. In particular, in full analogy to the $\beta = 0$ case, the full family can be given in terms of vectors tangent to the $n$-th linear fragment of the embedded majorisation curve,

\begin{align}
    \v{t}^{\mathfrak{M}(n)}(\v{p}) = \left(t_1^{\mathfrak{M}(n)}, \frac{p_n}{\gamma_n}\gamma^{\mathfrak{M}}_2, \hdots, \frac{p_n}{\gamma_n}\gamma^{\mathfrak{M}}_{2^d-1},t_{2^d}^{\mathfrak{M}(n)}\right) ,
\end{align}
where the first entry is defined in such a way that the majorisation curves, defined as the piecewise-linear functions given by their elbows $\left\{\left(\Gamma^{\mathfrak{M}}_i, P^{\mathfrak{M}}_i\right)\right\}_{i=0}^{2d-1}$, agree in at least one point, $f^{\mathfrak{M}}_{\v{t}^{(n)}}(\Gamma^{\mathfrak{M}}_n) = f_{\v{p}}^{\mathfrak{M}}(\Gamma^{\mathfrak{M}}_n)$ and the last one guarantees that $\sum_i t_i^{\mathfrak{M}(n)} = 1$. Observe that vectors $\v{t}^{\mathfrak{M}(n)}$ constructed in this way are tangent with respect to the embedded Lorenz curve, i.e., $t^{\mathfrak{M}(n)}_i/\gamma^{\mathfrak{M}}_i = t^{\mathfrak{M}(n)}_j/\gamma^{\mathfrak{M}}_j$, thus taking into account the varying intervals on the horizontal axis. Equipped with these, we pose a technical lemma similar to Lemma~\ref{lem-incomp},

As a preliminary step, we give the algorithm for the construction of the join in the embedding space by modifying the crucial steps \eqref{eq:joinConstruct_aCoeffs} and \eqref{eq:joinConstruct_flatFragments} to take into account the varying widths and redefine the point $N$ of increase by requiring $r_N^{\mathfrak{M}(k)}/\gamma_N^{\mathfrak{M}} > r_{N-1}^{\mathfrak{M}(k)}/\gamma_{N-1}^{\mathfrak{M}}$. Consequently, we redefine $M$ by a condition similar to \eqref{eq:joinConstruct_aCoeffs} that incorporates the scaling,
\begin{align} \label{finTemp_join_akSlope}
    \frac{r^{\mathfrak{M}(k)}_M}{\gamma^{\mathfrak{M}}_M} \geq \frac{\sum_{i = M}^N r_i^{\mathfrak{M}(k)}}{\sum_{i = M}^N \gamma_i^{\mathfrak{M}}} =: a^{\mathfrak{M}}_k.
\end{align}
Finally, we define the join candidate in the $k$-th step in full analogy to \eqref{eq:joinConstruct_flatFragments} as

\begin{equation} \label{finTemp_join_kStep}
    r_i^{\mathfrak{M}(k+1)} = 
    \begin{cases}
        a_k^{\mathfrak{M}}\gamma_i^{\mathfrak{M}} & \text{for } i\in\{M,\hdots,N\} \\
        r_i^{\mathfrak{M}(k)} & \text{otherwise}
    \end{cases}.
\end{equation}
The algorithm defined in this way follows precisely the same logic as the one given in Ref. \cite{Korzekwa2017} and thus it always terminates, in this in no more than $ d^\mathfrak{M}(\v p, \v q) - 1 = \overline{\left\{\sum_{i=1}^j \gamma_{\v \pi^{-1}_{\v{p}}(i)}\right\}_{j=1}^{d}\cup\left\{\sum_{i=1}^j \gamma_{\v \pi^{-1}_{\v{q}}(i)}\right\}_{j=1}^{d}} - 1$ steps.

As a final remark, one has to note that the family of tangent vectors $\v{t}^{\mathfrak{M}(n)}(\v{r}^\mathfrak{M})$ should be indexed by $n\in\{1,\hdots,d'\}$, where the number $d'$ of constant-slope fragments of the join, even though bounded, $d^\mathfrak{M}(\v p, \v q) \geq d' \geq d$, is \textit{a priori} not well defined due to many possible ways of disagreement between the $\beta$-orders of $\v{p}$ and $\v{q}$. With these tools, we are ready to present the technical lemma needed for constructing the incomparable region for $\beta > 0$.

\begin{lem}
\label{lem-incomp_nonZeroBeta}
$\mathfrak{M}(\v{q})\equiv\v{q}^\mathfrak{M}$ belongs to the incomparable region of $\mathfrak{M}(\v{p})\equiv \v{p}^\mathfrak{M}$, $\v{q}^\mathfrak{M} \nsucc_\mathfrak{M} \v{p}^\mathfrak{M}$, if and only if it belongs to the future thermal cone of some vector $\v{t}^\mathfrak{M}$ tangent to $\v{p}^\mathfrak{M}$, $\v{q}^\mathfrak{M}\succ\v{t}^\mathfrak{M}$, with
    $$
        \v{t}^\mathfrak{M} \equiv \v{t}^\mathfrak{M}(\v{p};\lambda, n)=  \lambda \v{t}^{\mathfrak{M}(n)}(\v{p}) + (1-\lambda) \v{t}^{\mathfrak{M}(n+1)}(\v{p}) ,
    $$
    for some $n\in \{1,d-1\}$ and $\lambda \in[0,1]$.
\end{lem}
\begin{proof}
    The proof follows in complete analogy with the standard majorisation case as presented in the proof for Lemma \ref{lem-incomp} by replacing the standard majorisation $\succ$ in every statement with the majorisation variant $\succ_\mathfrak{M}$ given for the embedding space and employing the adjusted join construction, summarised in equations \eqref{finTemp_join_akSlope} and \eqref{finTemp_join_kStep}. 
\end{proof}

In order to go back from the embedded space $\Delta_{d}^{\mathfrak{M}}$ to the formulation of Lemma \ref{lemma_incomparablefiniteT} in the original space $\Delta_d$ we combine two observations following from embedding and projection operations. First, note that majorisation in embedded space implies majorisation between projections, thus $\v{t}^\mathfrak{M} \succ_\mathfrak{M} \v{q} \Rightarrow \mathfrak{P}_{\v \pi}(\v{t}^\mathfrak{M}) \succ_\beta \mathfrak{P}_{\v \pi}(\v{q})$ for every order $\v \pi$. Second, observe that embedding preserves the majorisation relations between the vectors, therefore $\v{q}\in\mathcal{T}_\emptyset(\v p) \Leftrightarrow \mathfrak{M}(\v q) \in \mathcal{T}_\emptyset(\mathfrak{M}(\v p))$. These two statements show that it is enough to consider vectors $\v{t}^{(n,\v \pi)} = \mathfrak{P}_{\v \pi}(\v{t}^{\mathfrak{M}(n)})$ and convex combinations of their future thermal cones, thus proving Lemma~\ref{lemma_incomparablefiniteT}. \qed

\section{Construction of probabilistic majorisation cones}\label{app:probal_deriv}

\rd{
For the convenience of the reader, we restate the theorem concerning probabilistic transformations:

\vidalentanglement*

To have a direct connection to majorisation, we reformulate the above theorem as follows:
\begin{equation}
    \forall_{1\leq k\leq d}:\mathcal{P}(\v p, \v q) \leq \frac{\sum_{j = k}^d p^{\downarrow}_j}{\sum_{j = k}^d q^{\downarrow}_j} = \frac{1 - \sum_{j = 1}^{k-1} p^\downarrow_j}{1 - \sum_{j = 1}^{k-1} p^\downarrow_j} = \frac{1 - P_k}{1 - Q_k}.
\end{equation}
By setting $\mathcal{P}(\v p , \v q) = 1$, we recover the standard majorisation condition on deterministic convertibility,
\begin{equation}
    \forall_{1\leq k\leq d}: 1 \leq \frac{\sum_{j = k}^d p^{\downarrow}_j}{\sum_{j = k}^d q^{\downarrow}_j} \Leftrightarrow \v p \prec \v q.
\end{equation}
To determine the probabilistic past cone $\mathcal{T}_-(\v p, \mathcal{P})$ at probability $\mathcal{P}$, we consider 
\begin{equation}
    \begin{aligned}
       \forall_{1\leq k\leq d}: \mathcal{P} \leq \frac{1 - Q_k}{1 - P_k} 
        & \Rightarrow \mathcal{P} - \mathcal{P} P_k \leq 1 - Q_k \\
        & \Rightarrow Q_k \leq \mathcal{P} P_k + (1 - \mathcal{P}) \\
        & \Rightarrow \v{q}\prec\tilde{\v{p}},
    \end{aligned}
\end{equation}
with an auxiliary distribution
\begin{equation}
	\tilde{p}_i = \begin{cases}
		\mathcal{P} p^{\downarrow}_1 + (1 - \mathcal{P}) & \text{for } i = 1, \\
		\mathcal{P} p^{\downarrow}_i & \text{otherwise},
	\end{cases}
\end{equation}
which is always a proper probability distribution ordered non-increasingly, $\tilde{\v p}  = \tilde{\v p}^\downarrow$, therefore providing a proper Lorenz curve. 

Following a similar procedure for the future cone $\mathcal{T}_+(\v p, \mathcal{P})$ leads to
\begin{equation}
    \begin{aligned}
        \mathcal{P} \leq \frac{1 - P_k}{1 - Q_k} 
        & \Rightarrow \mathcal{P}^{-1} - \mathcal{P}^{-1} P_k \geq 1 - Q_k \\
        & \Rightarrow Q_k \geq \mathcal{P}^{-1} P_k + (1 - \mathcal{P}^{-1}) \\
        & \Rightarrow \v{q}\succ\hat{\v{p}},
    \end{aligned}
\end{equation}
with the second auxiliary distribution
\begin{equation}
	\hat{p}_i = \begin{cases}
		\mathcal{P}^{-1} p^{\downarrow}_1 + (1 - \mathcal{P}^{-1}) & \text{for } i = 1, \\
		\mathcal{P}^{-1} p^{\downarrow}_i & \text{otherwise}.
	\end{cases}
\end{equation}
In contrast to the case of the past cone, the distribution $\hat{\v p}$ in this formulation is not ordered non-increasingly beyond a certain value of $\mathcal{P}$. At first glance, one might think that reordering should solve the problem; however, it would be equivalent to a decrease in the probabilistic future with decreasing $\mathcal{P}$, which creates a contradiction. The solution is provided by noting that Vidal's criterion deals with rescaled entries of the Lorenz curve rather than the probabilities \textit{per se}. Therefore, the Lorenz curve for $\hat{\v p}$ should remain convex for all values of $\mathcal{P}$ without the need for reordering.

Consider the following critical values of $\mathcal{P}$, namely,
\begin{equation}
    \mathcal{P}_n = (n-1)p^{\downarrow}_n - \sum_{i=1}^{n-1} p^{\downarrow}_i + 1,
\end{equation}
for which the first $n$ entries of the distribution $\hat{\v{p}}$ will not be ordered non-increasingly, resulting in an improper Lorenz curve. The resulting non-convexity is controlled by replacing
\begin{equation}
    \left\{\hat{p}_1,\hdots,\hat{p}_n\right\} \rightarrow \frac{1}{n}\sum_{i=1}^n\hat{p}_i \left\{1,\hdots1\right\},
\end{equation}
which ensures that $\hat{\v p}  = \hat{\v p}^\downarrow$.

This way, the auxiliary ordered distributions $\tilde{\v{p}}$ and $\hat{\v{p}}$, together with the construction for the deterministic majorisation cones presented in Appendix \ref{sec_appendix}, provide the full construction of the probabilistic cones, with the additional cautionary note that the role of future and past majorisation cones is reversed when we consider entanglement and coherence theories.
}

\section{Volumes of entanglement majorisation cones} \label{app:entanglement_vols}

\rd{
    Consider a uniform Haar distribution of pure states in a composed space \mbox{$\ket{\psi}\in\mathcal{H}^{N}\otimes\mathcal{H}^M$} with $N \leq M$. The partial tracing induces a measure in the space of reduced states, $\rho = \operatorname{Tr}_2 \ket{\psi}\bra{\psi}$, characterised by the distribution of eigenvalues \mbox{$\Lambda = \left\{\lambda_1,\hdots,\lambda_N\right\}$} of the reduced state~\cite{zyczkowski2001induced}:
    \begin{equation}\label{eq_haardistribution}
        P_{N,M}(\Lambda) = C_{N,M} \delta\left(1 - \sum_i \lambda_i\right) \prod_i \lambda_i^{M-N}\theta(\lambda_i) \prod_{i<j} (\lambda_i - \lambda_j)^2,
    \end{equation}
where $\delta$ and $\theta$ are the Dirac delta and Heavyside step functions, respectively. The normalisation constant is given by
\begin{equation}
      C_{N,M} = \frac{\Gamma(NM)}{\prod_{j=0}^{N-1}\Gamma(M-j) \Gamma(N-j+1)},
\end{equation}
where $\Gamma$ is the Gamma function. Before we continue with our discussion, let us first understand the role played by all factors in Eq.~\eqref{eq_haardistribution}. As previously mentioned, $C_{N,M}$ is the normalisation. The delta function ensures that the spectrum sums to one (is normalised), whereas the step function will guarantee that it is positive. Now, notice that the first product (capital pi notation) essentially does not influence the behaviour of the distribution for \mbox{$N=M$} and introduces the repelling of the faces of the probability simplex otherwise, as it goes to zero whenever \mbox{$\lambda_i = 0$} for any~$i$. The second product is responsible for the repelling from distributions with any two entries equal, since it goes to zero whenever \mbox{$\lambda_i = \lambda_j$} for any~$i\neq j$.

\begin{figure*}
    \centering
    \includegraphics{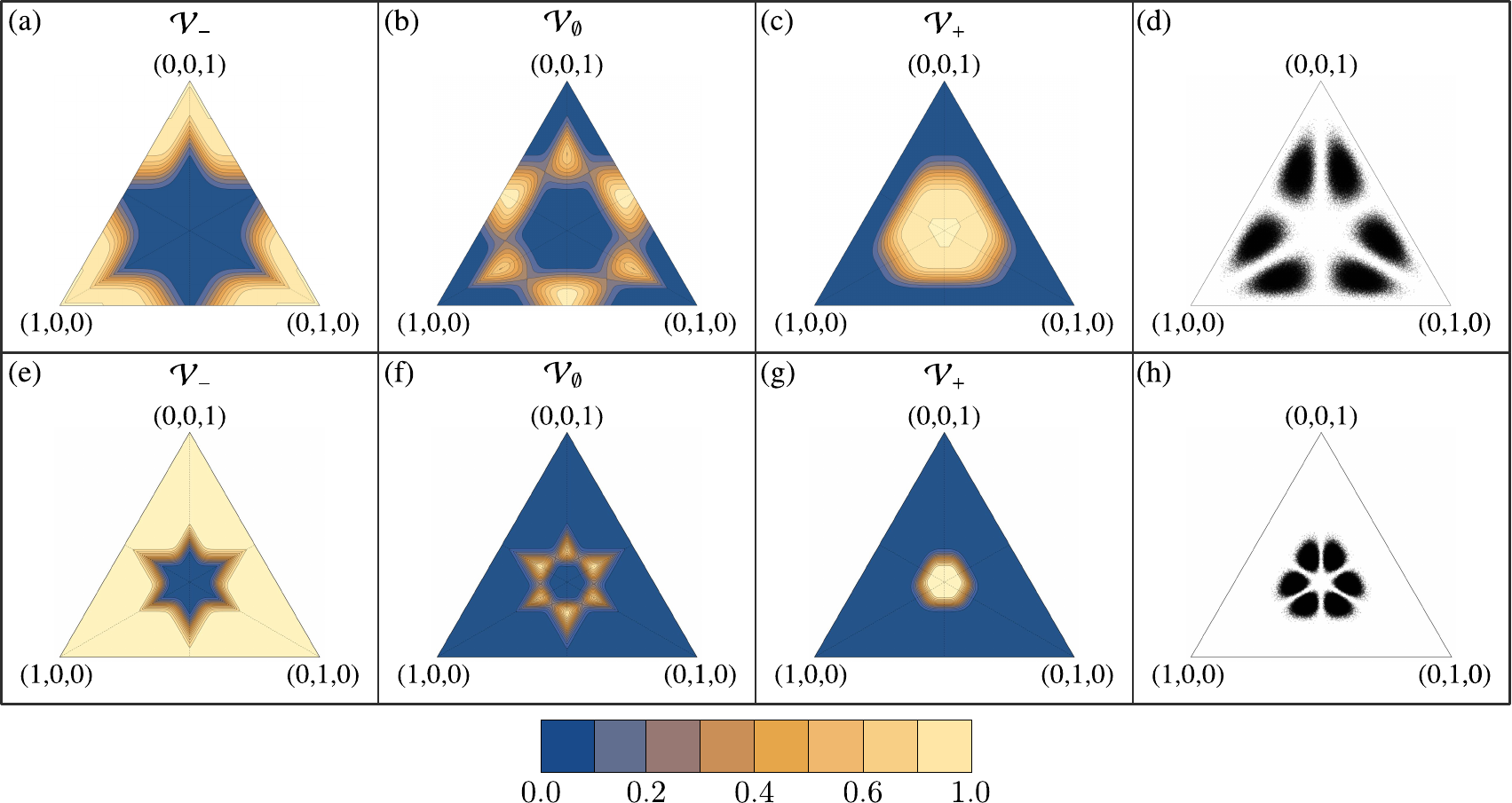}
    \rd{\caption{\label{fig-entanglementvolume-2} \emph{Isovolumetric sets for entanglement $3 \times M$ bipartite systems}. The density $P_{3,M}(\Lambda)$ of Schmidt coefficients of pure states for qutrit-quMit systems depends heavily on the dimension $M$ of the second system. Panels (a-d) and (e-h) present the isovolumetric lines for past, incomparable and future regions and the density of the states for $M = 6$ and $M = 30$, respectively. Note that for larger $M$ the density $P_{3,M}$ is more and more concentrated around the regions close to the centre [compare (d and (h)]. This affects the subset of states with large future volume, making it smaller [(c) and (g)] as well as the set of states with large past volume, enlarging it [(a) with (e)].}}
\end{figure*} 

Sampling from the $P_{N,M}$ distribution, ordinarily done by generating state vectors $\ket{\psi}\in\mathcal{H}_{NM}$ which would be computationally prohibitive for large $M$, can be achieved using only $O(N)$ random numbers for any dimension of the secondary system. It has been demonstrated in Ref.~\cite{Cunden_2020} that the distribution $P_{N,M}$ is precisely the Laguerre unitary ensemble generated by Wishart matrices of size $N$ and parameter $M$~ \cite{RandomMatrixBook} and, in turn, generated using a tridiagonal method containing only $O(N)$ random real numbers~ \cite{DE02}, which indeed allows one to study the $P_{N,M}$ distributions for arbitrary high-dimensional ancillary systems. 

The procedure for generating the isovolumetric lines for the majorisation cones with a given distribution, $P_{N,M}$, proceeds as follows:
    \begin{enumerate}
        \item Generate a sample of $n$ sets of eigenvalues $\{\Lambda_1,\hdots,\Lambda_n\}$ taken from the distribution $P_{N,M}$ using the tridiagonal method.
        \item Consider regularly spaced grid of points $S$ in a single chamber of the full probability simplex $\Delta_N$ (e.g., \mbox{$p_1 \geq p_2, ..., p_d)$} in order to avoid repeated counting (achieving $N!$ decrease in operations)
        \item For each $\v{p} \in S$ consider its majorisation cones $\mathcal{T}_i(\v p)$ and divide $S$ into $S_i \equiv \left\{\Lambda_i \in \mathcal{T}_i(\v{p})\right\}$.
        \item This way we arrive at the approximations of the volumes of the three regions, $$\mathcal{V}_i \approx \frac{|S_i|}{n},$$ where $|X|$ denotes the number of elements in a set $X$.
    \end{enumerate}

    We applied this method for $N=3$ with $M = 3$, displayed in the main text of Fig.~\ref{fig-level-sets-entanglement}, and additionally with $M = 6$ and $30$, as shown in Fig.~\ref{fig-entanglementvolume-2}. These two cases show the significant dependence of the isovolumetric lines on the size of the environment.
    }

\section{Coherent thermal cones for a two-level system}\label{app:coherent_thermal}

\rd{This appendix explains how to use the results of Refs.~\cite{LostaglioKorzekwaCoherencePRX}~and~\cite{Korzekwa2017} to construct the future and past thermal cones for a two-level system under thermal and Gibbs-preserving operations.}

\newpage

\subsubsection{Coherent thermal cones for thermal operations}

\rd{Consider initial and target states of a two level system, $\rho$ and $\sigma$, with both written in the energy eigenbasis as
\begin{equation}
\rho = \begin{pmatrix}
p & c \\ 
c & 1-p 
\end{pmatrix} \quad , \quad \sigma = \begin{pmatrix}
q & d \\ 
d & 1-q 
\end{pmatrix},
\end{equation}
where $c$ and $d$ are assumed to be real without loss of generality, which amounts to considering a cross-section of the Bloch ball in the $XZ$ plane. Moreover, the thermal ground state occupation of the considered two-level system will be denoted by $\gamma$. It has been shown that for thermal operations, the coherences of the initial and target states have to satisfy the following inequality \cite{LostaglioKorzekwaCoherencePRX},
\begin{equation}\label{eq_condition_coherence}
    d \leq c \frac{\sqrt{[q(1-\gamma)-\gamma(1-p)][p(1-\gamma) - \gamma(1-q)]}}{|p-\gamma|}.
\end{equation}
Thus, we find the boundary of the future thermal cone by saturating Eq.~\eqref{eq_condition_coherence}, and solving it for $q$ we obtain the achievable ground state occupation as a function of target coherence $d$, 
\begin{align}\label{eq_solution_coherence_fut}
q_1(d) &= \frac{(\gamma -p) \sqrt{c^2 (1-2 \gamma )^2+4\gamma d^2 (1-\gamma) }}{2\gamma  c (\gamma -1) } \nonumber \\ & \quad+\frac{(p-\gamma)-2\gamma  p(1-\gamma)}{2\gamma c (\gamma -1)}.
\end{align}
Therefore, the coherent future thermal cone is given by the region delimited by Eq.~\eqref{eq_solution_coherence_fut} from one side and a line segment connecting $(-c,p)$ and $(c,p)$. To characterise the coherent past thermal cone, it will be convenient to introduce a number $d_{\text{cross}} \geq 0$ defined by the relation \mbox{$d_{\text{cross}}^2 + q(d_{\text{cross}})^2 = 1$}.

The coherent past thermal cone is generically composed of two disjoint regions. The first region is contained between a line segment connecting the points $(c,p)$ and $(d_{\text{cross}},p)$, the curve $q_1(d)$ for $d\in[c, d_{\text{cross}}]$ and the boundary of the Bloch ball, together with its reflection with respect to the $Z$-axis.
The second one is obtained in a similar manner by focussing on the past state rather than the target, and thus by solving Eq.~\eqref{eq_condition_coherence} with interchanges $p\leftrightarrow q$ and $c\leftrightarrow d$. This results in
\begin{align}\label{eq_solution_coherence_pas}
q_2(d) &= \frac{2 \gamma c^2+\sqrt{c^2 (p-\gamma )^2 \left[(1-2 \gamma )^2 d^2-4 c^2 (\gamma -1) \gamma \right]}}{2 \left[d^2+(\gamma -1) \gamma  c^2\right]} \nonumber \\ & \:\:\:\:+ \frac{d^2 [p-\gamma -2\gamma  p (1-\gamma)]}{2 \left[d^2-(1-\gamma) \gamma  c^2\right]},
\end{align}
with $d \in [d_{\textrm{min}},d_{\textrm{max}}]$, where $d_{\textrm{min}}$ and $d_{\textrm{max}}$ are real positive solutions of equation $q(d)^2+d^2 = 1$, such that \mbox{$d_{\text{cross}} \leq d_{\text{min}} \leq d_{\text{max}}$}. However, for large values of coherence $c$, we note that this second region may not appear at all. Finally, the incomparable region $\mathcal{T}_{\emptyset}(\rho)$ is obtained by subtracting the past and future cones from the entire Bloch ball.

}

\subsubsection{Coherent thermal cones for Gibbs-preserving operations}

\rd{
Consider a parametrisation of qubit states $\rho$ in the Bloch sphere representation,
\begin{equation}
\label{eq:bloch_state}
\rho=\frac{\iden+\v{r}_\rho\cdot\v{\sigma}}{2},
\end{equation}
where \mbox{$\v{\sigma}=(\sigma_x,\sigma_y,\sigma_z)$} denotes the vector of Pauli matrices. The Bloch vectors of the starting state $\rho$, target state $\rho'$ and the Gibbs state $\gamma$ are given by:
\begin{equation}
\label{eq:bloch_param}
\v{r}_\rho=(x,y,z),\quad\v{r}_{\rho'}=(x',y',z'),\quad\v{r}_\gamma=(0,0,\zeta), 
\end{equation}
where the $z$ coordinate of the Gibbs state can be related to the partition function $Z$ by $\zeta=2Z^{-1}-1\geq 0$.

According to Ref.~\cite{Korzekwa2017}, there exists a GP quantum channel $\E$ such that $\E(\rho)=\rho'$ if and only if \mbox{$R_{\pm}(\rho)\geq R_{\pm}(\rho')$} for both signs, where \mbox{$R_{\pm}(\rho)=\delta(\rho)\pm\zeta z$} and 
\begin{equation}
	\label{eq:delta_lattice}
	\delta(\rho):=\sqrt{(z-\zeta)^2+(x^2+y^2)(1-\zeta^2)}.
	\end{equation}
Consequently, the future thermal cone $\T_+(\rho)$ of any qubit state $\rho$ under GP operations can be directly constructed from the above result. For a generic qubit state $\rho$, we first orient the Bloch sphere so that its $XZ$ plane coincides with the plane containing $\rho$ and a thermal state $\gamma$, i.e., \mbox{$\v{r}_\rho=(x,0,z)$}. Then, define two disks, $D_1(\rho)$ and $D_2(\rho)$ with corresponding circles $C_1(\rho)$ and $C_2(\rho)$, of radii 	
	\begin{equation}
	\label{eq:radii}
	R_1(\rho)=\frac{R_-(\rho)+\zeta^2}{1-\zeta^2},\quad R_2(\rho)=\frac{R_+(\rho)-\zeta^2}{1-\zeta^2},
	\end{equation}
centred at 
	\begin{equation}
	\label{eq:centres}
	\begin{array}{ccc}
		\v{z}_1(\rho)&=&[0,0,\zeta(1+R_1(\rho))],\\ \v{z}_2(\rho)&=&[0,0,\zeta(1-R_2(\rho))].
	\end{array}
	\end{equation}
Therefore, the future thermal cone under GP quantum channels is given by the intersection of two disks of radii $R_1(\rho)$ and $R_2(\rho)$ centred at $\v{z}_1(\rho)$ and $\v{z}_2(\rho)$, $\mathcal{T}_+(\rho) = D_1(\rho)\cap D_2(\rho)$.

The incomparable region is given by mixed conditions, i.e., $\rho'\in\mathcal{T}_\emptyset(\rho)$ if and only if $R_{\pm}(\rho)\geq R_{\pm}(\rho')$ and $R_{\mp}(\rho)< R_{\mp}(\rho')$, or in terms of the disks given beforehand, $\mathcal{T}_\emptyset(\rho) = D_1(\rho)\cap D_2(\rho) \backslash \mathcal{T}_+(\rho)$. Finally, the past cone $\mathcal{T}_-(\rho)$ can be easily given by subtracting the future cone and the incomparable region from the entire Bloch ball.

}

\bibliographystyle{apsrev4-2}
\bibliography{references}

\end{document}